\documentclass[]{IEEEtran}

\usepackage{graphicx}% Include figure files
\usepackage{dcolumn}% Align table columns on decimal point
\usepackage{bm}% bold math
\usepackage{amsmath}
\usepackage{amsthm}
\usepackage{amssymb}
\usepackage{amsfonts}
\usepackage{cite}

\usepackage{cite}
\usepackage{algorithm}
\usepackage{algpseudocode}
\usepackage{enumitem}
\usepackage{subcaption}
\newtheorem{assumption}{Assumption}
\newtheorem{corollary}{Corollary}
\newtheorem{definition}{Definition}
\newtheorem{theorem}{Theorem}
\begin{document}

%\preprint{APS/123-QED}

%\title{Fleming-Viot helps speed up variational quantum algorithms in the presence of barren plateaus}
%\title{Swarm of gradient-aware variational quantum algorithms to speed up optimisation in the presence of barren plateaus}
\title{Parallel variational quantum algorithms with gradient-informed restart to speed up optimisation in the presence of barren plateaus}

\author{Daniel Mastropietro, Georgios Korpas, Vyacheslav Kungurtsev, Jakub Marecek}

%\author{Daniel Mastropietro}
% \email{daniel.mastropietro@gmail.com}
%\affiliation{CNRS-IRIT, Universit\'e de Toulouse INP, 31071, Toulouse, France}
%\orcid{0000-0002-2445-2701}
%\author{Georgios Korpas}
%\orcid{0000-0003-0290-4698}
%\affiliation{Quantum Technologies Group, Emerging Technology, Innovation \& Ventures, HSBC, Singapore}
%\affiliation{Department of Computer Science, Czech Technical University in Prague, Czech Republic}
%\affiliation{Archimedes Research Unit on AI, Data Science and Algorithms, Athena RC, Athens, Greece}
%\author{Vyacheslav Kungurtsev}
%\affiliation{Department of Computer Science, Czech Technical University in Prague, Czech Republic}
%\orcid{0000-0003-1985-4623}
%\author{Jakub Marecek}
%\affiliation{Department of Computer Science, Czech Technical University in Prague, Czech Republic}
%\orcid{0000-0003-0839-0691}

\date{\today}% It is always \today, today,

\maketitle

\begin{abstract}
  Inspired by the Fleming-Viot stochastic process, we propose a parallel implementation of variational quantum algorithms with the aim of reducing the time spent by the algorithm in
  barren plateaus, where optimization direction is unclear.
  In the Fleming-Viot tradition, parallel searches are called particles. In the proposed approach, the search by a Fleming-Viot particle is stopped when it encounters a region where the gradient is too small or noisy, suggesting a barren plateau area. 
  The stopped particle continues the search after being regenerated at another location of the parameter space, potentially taking the exploration away from barren plateaus.
  We first analyze the behavior of the Fleming-Viot particles from a theoretical standpoint.
  We show that, when simulated annealing optimizers are used as particles,
  the Fleming-Viot system is expected to find the global optimum faster than a single simulated annealing optimizer,
  with a relative efficiency that increases proportionally
  to the percentage of barren plateaus in the domain.
  This result is corroborated by numerical experiments carried out on synthetic problems as well as on instances of the Max-Cut problem,
  which show that our method performs better than plain
  simulated annealing
  when large barren plateaus are present in the domain.

    %\begin{description}
    %\end{description}
\end{abstract}

%\keywords{Max-Cut, Parallelization, Swarm optimisation, Simulated annealing}
%\keywords{Fleming-Viot, Max-Cut, Parallelization, Simulated annealing}

\maketitle

%\tableofcontents

\section{Introduction}

In both academia and industry, there has been considerable interest in solving optimization problems using near-term quantum computers \cite{Preskill_2018,survey}, %albeit often in the form of quantum annealers \cite[e.g.]{kadowaki2002study,2005, yarkoni2022quantum,mehta2021quantum,Crosson2021,Chakrabarti2022,certo2023predicting}, or in the 
often using variational quantum algorithms (VQAs) on gate-based quantum computers \cite{farhi2014quantum,mcclean2016theory,Guerreschi2019,cerezo2021variational,egger2021warm}. % or some combination of the two \cite[e.g.]{sack2021quantum}.  
VQAs \cite{mcclean2016theory, cerezo2021variational} are sometimes suggested as potential candidates to demonstrate an early quantum advantage in the near-term noisy quantum (NISQ) computers due to their simplicity, robustness to noise, and because the resources required are shared between a quantum computer and a classical optimizer. 
So far, VQAs have shown some potential in applications in a number of domains, including chemistry \cite{o2016scalable,kandala2017hardware,Colless2018,mccaskey2019quantum,mcardle2020quantum,google2020hartree}, strongly correlated systems in condensed matter physics \cite{kokail2019self,lyu2020accelerated,lau2021noisy,haug2022generalized,gong2023quantum}, combinatorial optimization \cite{farhi2014quantum,farhi2016quantum,wang2018quantum,crooks2018performance,moll2018quantum,harrigan2021quantum,brandhofer2022benchmarking,boulebnane2022solving,PhysRevResearch.4.033029}, solving
linear and nonlinear equations \cite{bravo2023variational,xu2021variational,PhysRevA.101.010301}, supervised quantum machine learning \cite{schuld2019quantum,havlivcek2019supervised}, generative models \cite{dallaire2018quantum,benedetti2019generative,PhysRevResearch.4.043092}, and quantum neural networks \cite{du2020expressive, NEURIPS2020_0ec96be3, cherrat2024quantum, hou2023duplication}, to name a few.

VQAs are iterative algorithms, where each iteration runs a parameterized quantum circuit that encodes the problem of interest, measures the expected value of some observable relevant to this problem, and classically learns the parameters that minimize a cost function related to the expected value of the observable. These iterations are repeated until the process converges to some satisfactory value. See Fig. \ref{fig:myvqa} overleaf for an illustratioon.

\subsection{Challenges encountered in VQAs}
Despite the potential of variational quantum algorithms, there exist several challenges in
their deployment to practical problems.  Concretely, the optimization problems involved in classically learning the parameters are non-convex, and it is expected that classical subroutines require super-exponential time in the number of variables to find the global optimum. In practical implementations, one often utilizes gradient-based classical optimization subroutines 
%unless some underlying tame structure exists \cite{bolte2021conservative,NEURIPS2021_70afbf22,aravanis2022polynomial,aspman2023riemannian,aspman2024taming}, 
with no guarantee that such a subroutine will converge to a global minimum
%whatsoever 
in a finite amount of time \cite{bolte2021conservative,NEURIPS2021_70afbf22,aravanis2022polynomial,aspman2023riemannian,aspman2024taming}.
Specifically, the gradient-based algorithm may get trapped in (suboptimal) local minima or saddle points.
%In \cite[Theorem 1.1]{you2022convergence} it is shown that in the overparameterized regime \textsc{poly}$(n,\kappa)$ where $n$ is the number of qubits and $\kappa = \tfrac{\lambda_n-\lambda_1}{\lambda_2-\lambda_1}$, for $\lambda_i$ the $i$-th eigenvalue of the problem Hamiltonian. Furthermore,
The iteration complexity, as well as the impact of the biased noise present in NISQ devices, are studied in \cite{kungurtsev2024iteration}. Further convergence results are presented in \cite{binkowski2024elementary}.

\begin{figure}[!ht]
    \centering
    \includegraphics[scale=0.9]{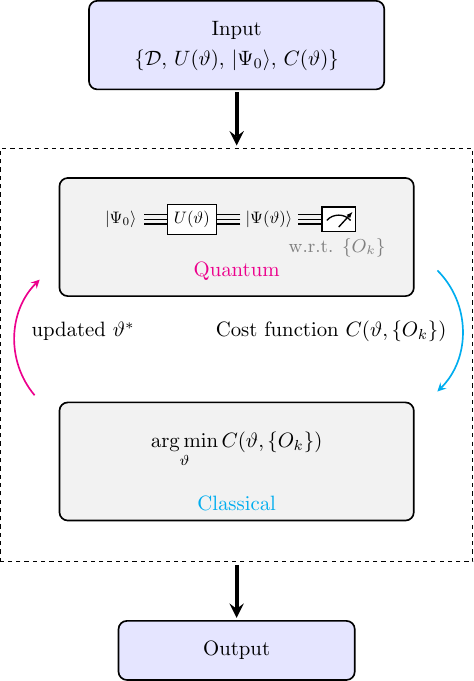}
    \caption{The VQA schema, as presented in \cite{survey}. In dotted square the classical-quantum feedback loop. The circuit is initialized with parameters $\vartheta$ with the purpose of obtaining the empirical estimate of the expectation value of a set of observables $\{ O_k \}_{k\in \mathbb{N}}$. In the classical part, a classical optimizer minimizes the cost function in order to obtain a new set of parameters, $ \vartheta^*$ to be fed back to the quantum circuit.}
    \label{fig:myvqa}
\end{figure}

One major drawback of variational quantum algorithms is that they suffer from the so-called barren plateau  problem \cite{mcclean2018barren,wang2021noise,Cerezo2021other,arrasmith2021effect,holmes2021barren,larocca2022diagnosing},  especially when the problem is approached with hardware efficient ans\"atze.
%(that is, initial solution guesses $\vartheta$ of the problem). 
Barren plateaus are large regions in the parameter space where the gradient of the cost function vanishes or is noisy around zero, which, in turn, makes efficient optimization very difficult. (See Def.~\ref{def:barren_plateau} on page \pageref{def:barren_plateau} for a formal definition.) Specifically,  gradient descent iteration is:
\begin{align*}
    \vartheta(t+1)-\vartheta(t) =-\eta \nabla f,
\end{align*}
where $t$ is the integer time step of gradient descent dynamics, $\eta$ is the learning rate, and $f$ is the cost function. The observation of Ref. \cite{mcclean2018barren} is that, if the chosen variational circuit is highly random, the gradient of the cost function is generically suppressed by the dimension of the Hilbert space, and the variation of the cost function during gradient descent is very small, that is, $\delta f \equiv \frac{f(t+1)-f(t)}{f(t)} \ll 1$ for step $t$.
A Lie-algebraic characterization of barren plateaus can be used to show that the gradient vanishes exponentially fast with the number of qubits, $n$. Specifically,
%since the generators of a quantum circuit (i.e., the three Hermitian Pauli matrices) are proportional to the generators of the universal Lie algebra $\mathfrak{su}(2)$, a quantum system with $n$ qubits can be represented by $\mathfrak{su}(2^n)$. Using this, the authors of
\cite{fontana2024characterizing} proved that, in the worst-case scenario, the gradient variance vanishes with the inverse of the dimension of the Lie algebra corresponding to the circuit,
%the dependence of the vanishing gradients on the dimension of the underlying Hilbert space,
%more specifically,}
as
\begin{align*}
    %{\rm Var} (\partial_{\vartheta_\ell} f) = \mathcal{O}(2^{-2n}).
    {\rm Var} (\|\nabla f\|) = \mathcal{O}(2^{-2n}).
\end{align*}
In particular, this worst-case behaviour can be demonstrated \cite{Sauvage2024} in the fully connected Max-Cut problem analyzed in this work.

In general, the problem of barren plateaus has also been argued to be related to the choice of cost function \cite{cerezo2021cost}, circuit entanglement \cite{patti2021entanglement,marrero2021entanglement}, circuit parameterization expressivity \cite{holmes2022connecting}, the presence of Pauli noise \cite{wang2021noise} while it has been suggested that gradient-free methods such as the finite difference stochastic approximation algorithm (FDSA) and the simultaneous perturbation stochastic approximation algorithm (SPSA) may also suffer from the problem of barren plateaus \cite{arrasmith2021effect}. It should be noted that, within quantum machine learning applications, similar to deep learning, one wants to achieve a high level of expressivity. Recent work in Ref. \cite{schumann2024emergence} has shown that highly expressive parameterized circuits (deep circuits) give rise both to standard as well as noise-induced barren plateaus. In particular, \cite{mele2024noiseinduced} shows that the effectiveness of deep circuits is reduced by the presence of noise, suggesting that shallower circuits would be preferred thanks to their lower noise. Therefore, we conclude that this barren plateau problem can hinder the possibility of observing near-term advantage with variational quantum algorithms.

Several methods have been devised to avoid barren plateaus. For example, \cite{PRXQuantum.3.020365} uses the technique of classical shadows \cite{huang2020predicting}. Other approaches have been proposed in several works, for example, \cite{zhang2022escaping,friedrich2022avoiding}. Worth mentioning is the approach of Ref.\ \cite{binkowski2023barren} where the authors construct certain conic extensions of the parameterized quantum circuits allowing for jumps out of barren plateaus. Other methods include certain parameter initialization
strategies \cite{grant2019initialization}, gradients via parameter correlation \cite{volkoff2021large}, ansatz classical neural network pre-training
\cite{verdon2019learning}, and layer-wise training \cite{skolik2021layerwise}.

%In parallel, an approach to avoid noisy gradients is DISQ \cite{zhang2023disq}. Specifically, it is a method of crafting a stable landscape for VQA training and tackling the challenge of noise drift. DISQ adopts a ``drift detector'' with a reference circuit to identify and skip iterations that are severely affected by noise drift errors. Concretely, the circuits from the previous training iteration are re-executed as a reference circuit in the current iteration to estimate noise drift impacts. The iteration is deemed compromised by noise drift errors and thus skipped if noise drift flips the direction of the ideal optimization gradient. To enhance noise drift detection reliability, we further propose to leverage multiple reference circuits from previous iterations to provide a well-founded judge of current noise drift. However, multiple reference circuits also introduce considerable execution overhead. To mitigate extra overhead, we propose Pauli-term subsetting (prime and minor subsets) to execute only observable circuits with large coefficient magnitudes (prime subset) during drift detection. Only this minor subset is executed when the current iteration is drift-free. Evaluations across various applications and QPUs demonstrate that DISQ can mitigate a significant portion of the noise drift impact on VQAs and achieve 1.51-2.24x fidelity improvement over the traditional baseline. DISQ's benefit is 1.1-1.9x over the best alternative approach while boosting average noise detection speed by 2.07x

Before we dive into our proposal, we would like to motivate our ideas by considering a problem that appears to be quite similar. Consider random walks on graphs. Despite their natural appeal for graph exploration and numerous applications in distributed systems, random walks exhibit a significant limitation: slow mixing. For example, in a simple ring topology with $n$ nodes, a random walk requires an average of $\mathcal{O}(n^2)$ steps to cover the entire graph, whereas a deterministic traversal can achieve the same in just $ n $ steps. The cover time, which represents the expected duration for a random walk to visit every node at least once, emerges as a critical metric, emphasizing the challenges posed by random walks. Ref. \cite{alon2008many} suggested a way to address the latency of random walks by utilizing multiple simultaneous random walks. The hypothesis is that by initiating $k$ random walks concurrently, the graph coverage can be expedited, reducing the cover time. It has been empirically shown that a spectrum of outcomes, ranging from logarithmic speed-ups in cases like ring graphs to exponential speed-ups in structures like the barbell graph, can be obtained. However, for a wide range of interesting graphs, only a linear speed-up in cover time has been observed, as long as $k$ remained within a permissible range. 

Interestingly, the above limitation is somewhat related to the training of a variational quantum algorithm as seen from the point of view of a non-finite unbounded graph. In particular, the parameter landscape is analogous to intricate graph structures where random walks struggle with the extra difficulty of the vanishing gradients. 

\subsection{Our contribution}
The usage of multiple parallel random walks to expedite coverage \cite{MARTI20131} can inspire similar potentially heuristic strategies in this context and to some extent our proposal follows along this line. As we will explain in detail in Sec.~\ref{sec:fv}, by employing multiple parallel gradient descent initializations with a specific killing and regeneration rule, the so-called Fleming-Viot inspired algorithm, described in Alg.~\ref{alg:FV} of App.~\ref{app:algorithm}, seems to better navigate the parameter landscape, efficiently circumventing barren plateaus by
%biasing
taking the paths towards other locations that potentially avoid them.

The rest of the paper is organized as follows: in Sec.~\ref{sec:fv} we introduce the Fleming-Viot process and describe how we can apply its principles to the classical optimization step of VQAs, which gives rise to Alg.~\ref{alg:FV}.
In Sec.~\ref{sec:theorem} we provide and prove a theorem on the reduced expected hitting time of our method against vanilla simulated annealing.
In Sec.~\ref{sec:num_exp} we showcase an application of our Fleming-Viot inspired algorithm for a synthetic problem with varying percentages of barren plateaus in the parameter landscape, as well as an application to a Quantum Approximation Optimization Algorithm (QAOA, a special case of VQA) for an instance of the Max-Cut problem on a graph with 8 nodes. We summarize and conclude in Sec.~\ref{sec:sumconc}.

\section{Fleming-Viot Process within VQAs}\label{sec:fv}

The Fleming-Viot stochastic process (FV) was initially introduced by Fleming and Viot \cite{flemingviot1979}, to model biological evolution. In the 2000s, several works studied the relationship between the Fleming-Viot process and Markov processes, starting with Burdzy \emph{et al.}
\cite{Burdzy2000} who proved it can be used to approximate the eigenfunctions of the Laplacian (the generator of a Brownian motion) in a bounded domain.
More recently, \cite{grigorescu2004hydrodynamic}
%, \cite{grigorescu2006taggedparticle}, \cite{grigorescu2012immortal}
proved this approximation is valid for diffusion processes that are more general than the Brownian motion, in both bounded and unbounded domains. The FV process was originally defined as a continuous-time process, but it can also be defined as a discrete-time process \cite{budhiraja2022approximating}, which is the context that will be utilized in this work.

In its original definition, the Fleming-Viot process is a continuous-time stochastic process consisting of $N$ particles that initially evolve independently of each other on a state space $\Lambda \cup \{0\}$. Each particle follows the dynamics of an underlying Markov process as long as it does not visit $0$, an absorption state, i.e., where the underlying process stays forever or, in other words, is killed. The FV process is defined so that when a particle is killed, one of the remaining particles, chosen uniformly at random, is split into two particles, each at the same state in $\Lambda$ of the particle before the split. After the split, each particle evolves independently of each other, just like before the split.

In the context of the present article, the Fleming-Viot process is used as a driving idea of a parallelized optimization process. Each particle represents a different instance of a stochastic optimizer of choice (e.g., SPSA or simulated annealing) of a parameter vector of interest, $\vartheta$, whose values evolve following the rules of the selected optimizer. Each particle initially evolves independently of each other until it is considered killed, i.e. when the corresponding $\vartheta$ reaches a subset of the parameter space considered not to be interesting for further exploration by the optimization process. At this point, two things can happen to the killed particle, which correspond to two different regeneration schemes: either the particle is reactivated, which corresponds to an exploitation strategy as it exploits the information about other particles, or the particle is reinitialized, which corresponds to an exploration strategy as it favours the exploration of the parameter space. More specifically,
\begin{itemize}
    \item under the exploitation strategy, the particle is regenerated to the position of another particle chosen uniformly at random among the particles not currently killed --which are expected to be in a more interesting location than the killed particle.
    \item under the exploration strategy, the particle is regenerated to a randomly chosen point in the parameter space. 
\end{itemize}
Fig.~\ref{fig:fv_particle_regeneration} illustrates this. 
For every killed particle, the choice of one of these options is made with a fixed probability defined by the so-called ``exploration rate'' parameter.
Notice that in the first case, and thanks to the stochastic nature of the optimizer, the two particles initially at the same position evolve independently of each other after the regeneration.

Overall, in the VQA context, the FV parallel optimizer is used for the classical optimization part with a ``killing'' rule on the average magnitude of the gradient computed on a recent window of optimization steps.

%\begin{figure}[htb]
%    \centering
%    \includegraphics[scale=0.5]{figures/particles.pdf}
%%    \caption{}
%  \label{fig:fv}
%\end{figure}

The details of the FV parallel optimization algorithm are defined in Alg.~\ref{alg:FV} of App.~\ref{app:algorithm}.

\section{Theoretical Analysis}\label{sec:theorem}

Our main result is a hitting-time analysis for Alg.~\ref{alg:FV}.
In the next section (Sec.~\ref{sec:num_exp}), the analysis is corroborated by numerical results.
To present the analysis, let us define several quantities we will refer to:
\begin{itemize}
    \item $\mathbb{X}$ is a normed, finite-dimensional space. 
    \item $x \in \mathbb{X}$ will denote the optimization variable, i.e. $\vartheta$ in Fig.~\ref{fig:myvqa}. 
    \item $X_k$ denotes the random process associated with the iterates at optimization step $k$of the two analyzed algorithms, simulated annealing with Fleming-Viot (SA+FV) and without (SA).
    \item $D$ is the (Lebesgue) measure of the domain for the optimization variable $x$.
    \item $f(x)$ is the objective function with Lipschitz-smooth gradient.
\end{itemize}
We denote the equation corresponding to SA as:
\begin{equation}\label{equ:sa_iter}
X_{k+1}=X_k-\eta_k g_k+\sqrt{\eta_k}\epsilon_k n_k
\end{equation}
where $g_k\approx \nabla f(x_k)$ is a noisy estimate of the gradient (e.g. given by SPSA), $\eta_k$ is the learning rate, $\epsilon_k$ is a temperature factor defining the noise level of the SA algorithm, and $n_k$ is a standard Gaussian random variable.
We are interested, in particular, with the impact of the barren plateaus on the hitting time of SA \eqref{equ:sa_iter}, where:

\begin{definition}[Barrren plateau of a function $f$]
\label{def:barren_plateau}
Let us have $\alpha > 0$ and a ball $\mathbb{B}(\hat{x},R) \subset \mathbb{X}$ of radius $R$ centered at $\hat{x} \in \mathbb{X}$.
The ball $\mathbb{B}(\hat{x},R) $ is an $\alpha$-barrren plateau of function $f$ whenever 
$\|\nabla f(x)\|\le \alpha$ for all $x\in \mathbb{B}(\hat{x},R)$.
\end{definition}

Furthermore, we assume:

\begin{assumption}[Stationary points and regions]
\label{ass:1}
    When we observe the stochastic process \eqref{equ:sa_iter} arriving at $\| X_{k+1} - X_k \|_2^2 \le \beta$,  we assume that for some predefined small positive constant $\beta$, this indicates that $X_{k}$ satisfies $\|\nabla f(X_{k})\|\le \alpha$. Furthermore we assume that
    this is a local minimizer with probability $l/s$,
    a saddle point with probability $a/s$,
    and an $\alpha$-barrren plateau with probability $b/s$.
    Simplistically, this could be seen as having 
    $s=l+b+a$ of $\beta$-stationary points or regions, out of which there are 
    $l$ local minimizers, one of which is global, $b$ plateaus that are $\alpha$-barrren, and $a$ saddle points. 
    At the same time, consider that a fraction $0<B<1$ of the space is barren, i.e. $BD$ is the total measure of a barren plateau, and that each $\alpha$-barrren plateau is a ball of equal size, i.e., of $BD/b$ measure, from which we can compute the radius to be proportional to $(BD/b)^{1/d}$, with $d$ the dimension of $x$. %Clearly, $b$ and $B$ are not independent. 
\end{assumption}

When $x$ is in an $\alpha$-barren plateau, the noise in the iteration dominates the gradient term. Thus, we can treat the gradient as insignificant, and so we simplistically model the processes.

\begin{assumption}[Behaviour in an $\alpha$-barren plateau]
\label{ass:2}
    We assume that with FV, identification of barrenness is performed with full accuracy. 
    Thus, upon reaching an $\alpha$-barren plateau:
    \begin{enumerate}
        \item SA+FV immediately identifies the plateau and performs a jump escaping the plateau.
        \item SA behaves as a Gaussian random walk, being defined entirely by the noise component of the iteration.
    \end{enumerate}
\end{assumption}
    
Then:

\begin{theorem}[Hitting-time analysis]
    Under Assumptions~\ref{ass:1}--\ref{ass:2}, the expected time of SA+FV to hit a global minimizer is given by the sum of a time to hit a stationary point $C_s D\alpha^{-4}$ and 
    \[%C_s D\alpha^{-4}+
    (s-1)\left[C_s D\alpha^{-4}+\frac{a C_a}{s\alpha}+ \frac{(l-1) C_l}{s\alpha} \right]
    \]
    %\[ s C_s D\alpha^{-4}+(s-1)\left[\frac{a C_a}{s\alpha}+ \frac{(l-1) C_l}{s\alpha} \right]
    %\]   
    while the expected time of SA \eqref{equ:sa_iter} without FV to hit a global minimizer is again the sum of the time to hit a stationary point $C_s D\alpha^{-4}$ and 
    \[
    \begin{array}{l}
    (s-1)
    \left[C_s D\alpha^{-4}+\frac{a C_a}{s\alpha}+ \frac{(l-1) C_l}{s\alpha}
    %\right.\\\qquad\qquad \left.
    + \mathbf{\frac{C_fBD}{sb}}\right]     
    \end{array}
    \]
    where the constants $C_s, C_a, C_l, C_f$ scale the escape time for the stationary points of Assumptions~\ref{ass:1}, that is,
    stationary points, saddle points, local minimizers, and barren plateaus. 
\end{theorem}
\begin{proof}
 Consider a process $X_k$ driven by SA with FV and $\hat{X}_k$ driven by SA without it. Consider the beginning of an event as the initialization of SA, at $k=0$, at some random non-barren point.

With $\alpha$ the gradient threshold and a decreasing learning rate $\eta_k \sim k^{-1/2}$, the hitting time to a stationary point is $C_s D\alpha^{-4}$    \cite[Theorem 3.3]{chen2020stationary}. We consider that generally there is an equal probability of the iterate hitting any one of the $s$ stationary points. Let us consider each type of point and its consequences on the two algorithms' iterations.

With probability $1/s$, the process has reached a global minimum.

Otherwise, with probability $a/s$, both algorithms take $C_a\alpha^{-1}$~\cite[Lemma 4.2]{chen2020stationary} to escape a saddle point. Subsequently the next iterated event begins, i.e., we have another SA path from an agnostic starting point proceeding until it hits a stationary point. 

Alternatively, with probability $(l-1)/s$, both algorithms take $C_l\alpha^{-1}$~\cite[Corollary 6]{zhang2017hitting} to escape a local minimum. Subsequently, the next iterated event begins.

Finally, with probability $b/s$, the algorithm reaches a barren plateau. At this point FV escapes, reinitializing the next event (to a randomly chosen point).
However, without FV, the time the algorithm remains at the barren plateau has the expected value of the sojourn time of a Gaussian random walk in a convex body.
The expected sojourn time is given by $C_f R^d$~\cite{ciesielski1962first}. Subsequently, the event restarts. 

Now we sum these quantities.
The final run, that is, the event at which the stationary point reached is a global minimizer, lasts for as long as the algorithm iterates reach a stationary point (which in this case is the global minimum).

The expected number of stationary points traversed before a global minimizer is reached is a geometric random variable with parameter $1/s$. Thus, for an expected $s-1$ time steps, we hit the other stationary points, and this enables us to compute the expected hitting time as the $s-1$ sum of the probabilities of hitting the different types of stationary points and their expected contribution to the total iteration time.
\end{proof}

This means that:

\begin{corollary}
Compared to the expected hitting time of SA, the expected time of SA+FV to hit a global minimizer is sped up
by a total expected number of iterations given by 
$$\frac{C_fBD(s-1)}{bs}$$
i.e., the speed up is an additive term that is linear with respect to both the number of separated barren plateaus as well as the portion of the space that is barren.
\end{corollary}
Note that this appears to be independent of $\alpha$. However, this is misleading, as we expect the definition of the threshold to affect $B$ and $b$, thus the more relaxed the condition, the larger $\alpha$, the larger the region that is labeled ``barren''. 

\begin{figure}[tb]
    \centering
    \includegraphics[scale=1]{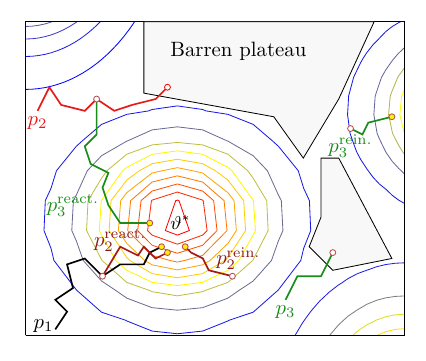}
    \caption{Illustration of the FV optimization algorithm using $N = 3$ particles, showing the two options described in Sec.~\ref{sec:fv} of particle regeneration after it is considered absorbed or killed: (i) particle reactivation (a.k.a. exploitation strategy), where an absorbed particle is regenerated to the position of one of the other particles selected uniformly at random, and (ii)
    particle reinitialization (a.k.a. exploration strategy),
    where an absorbed particle is regenerated to a randomly selected point in the parameter space, according to some probability distribution (e.g. uniform distribution). In the picture, particle $p_1$ is never absorbed, while particles $p_2$ and $p_3$ are absorbed at two different barren plateaus, for which both a reactivation and reinitialization paths are shown.}
    \label{fig:fv_particle_regeneration}
\end{figure}

\section{Numerical Experiments}\label{sec:num_exp}
We analyzed the FV parallel optimizer using, as each particle's optimizer, the simulated annealing (SA) algorithm, which sits on top of the simultaneous perturbation stochastic approximation (SPSA) algorithm \cite{Pincus1970,Khachaturyan1981,vanLaarhoven1987}. This choice was driven by the fact that the SA algorithm is a good balance between simplicity and power.

Given $N$ particles in the FV optimizer and $T$ optimization iterations for each particle, optimal values found on selected optimization problems are compared among four different optimization strategies:
\begin{enumerate}
\item ``FV-50\%'': FV optimizer with 50\% exploration rate;
\item ``FV-100\%'': FV optimizer with 100\% exploration rate;
\item ``SA-1'': a single SA optimizer run for $N \times T$ optimization iterations;
\item ``SA-N'': $N$ SA optimizers starting at the same ans\"atze (initial parameter values) as those chosen for the particles in the FV optimizers.
\end{enumerate}

Two exploration rates are considered for the FV optimizer in order to analyze the effect of the exploitation vs. exploration strategies, which are controlled precisely by the exploration rate parameter. Note that, in the $100\%$ exploration rate, all killed particles are regenerated under the exploration strategy (described in Sec.~\ref{sec:fv}). 

In ``SA-1'', the number of iterations is $N$ times larger than the other methods so that it receives the same number of opportunities to find the global minimum as the other three methods. In ``SA-N'', the same ans\"atze as FV are chosen in order to reduce sources of variability in their performances and achieve a more robust comparison.
%, so that both optimizers start at the same locations in the parameter space.

We carried out the above optimization strategies on the following two non-convex optimization problems on a bounded domain:
\begin{enumerate}
    \item[(\emph{i})] two-dimensional synthetic functions generated with a predefined percentage of barren plateaus;
    \item[(\emph{ii})] the Max-Cut problem approximated by a noiseless two-dimensional quantum circuit.
\end{enumerate}

The FV parameters that are common to both problems are listed in Tab.~\ref{tab:parameterization_common}.

\begin{table}[htbp]
\centering
\scalebox{1.0}{
\begin{tabular}{|c|c|c|}
\hline
\multicolumn{1}{|c|}{\textbf{Parameter}}                        & \multicolumn{1}{l|}{\textbf{Value range}} & \multicolumn{1}{l|}{\textbf{Value(s) used}} \\ \hline
\begin{tabular}[c]{@{}c@{}}$N$: No. particles\end{tabular}       & $\geq 2$                                  & 10                                       \\ \hline
\begin{tabular}[c]{@{}c@{}}$B$: Burn-in time steps \\
for reference gradient\end{tabular}& $\geq 1$         & 10                                        \\ \hline
\begin{tabular}[c]{@{}c@{}}$W$: Window size for \\
average gradient\end{tabular}& $\geq 1$         & 10                                        \\ \hline
\begin{tabular}[c]{@{}c@{}}$\alpha$: Relative gradient \\
absorption threshold\end{tabular}& $\geq 0$         & 1                                        \\ \hline
\begin{tabular}[c]{@{}c@{}}$\epsilon$: Exploration rate\end{tabular}        & 0\%-100\%                                 & 50\%, 100\%                          \\ \hline
\end{tabular}
}
\caption{FV parameters used in both the synthetic and Max-Cut experiments, as defined in Alg.~\ref{alg:FV} of App.~\ref{app:algorithm}.}
    \label{tab:parameterization_common}
\end{table}

During the optimization procedure, given the scale of the parameter space in both problems ($\sim 1$), each parameter update given by Eq.~\eqref{equ:sa_iter} is limited to $0.1$, in order to avoid excessively large one-step updates. Any value that happens to fall outside the parameter domain in a given dimension is clipped to the domain border of the dimension. The SA temperature factor in  Eq.~\eqref{equ:sa_iter} is set to $\epsilon_k = 0.1$, $\forall k$; the learning rate $\eta_k$ is initially selected (automatically by the SPSA algorithm) from an analysis of the function's landscape in the neighbourhood of the ansatz, and then decayed with a $\frac{1}{k}$ power law, in order to satisfy the convergence conditions of the SPSA algorithm on which SA is based.

Each set of experiments is defined by a parameterization specific to the optimization problem under consideration, described in the following two sections. Each experiment was run for $T = 50$ optimization steps per particle, and the results are jointly presented in Sec.~\ref{sec:results}.

\subsection{Synthetic Experiments}\label{sec:num_exp_synthetic}
Each synthetic experiment is defined by an expected percentage of barren plateaus in the optimization landscape whose domain is set as the two-dimensional rectangular area $[0, 1] \times [0, 1]$.
The synthetic function is constructed as the output of a smoothing spline interpolation on a set of $(x, y, z)$ points, using the procedure described in App. \ref{app:synthetic_function_generation}. The synthetic functions generated in this way are noiseless and smooth.

With the goal of understanding the behaviour of the optimization methods on different barren levels, synthetic functions are generated using three different desired nominal barren percentages: 25\%, 50\%, 80\%, aiming at generating function instances with low, medium and high barren levels. The precise definition of these three barren levels is given in Section~\ref{sec:results}. Each accepted\footnote{As described in App.~\ref{app:synthetic_function_generation}.} synthetic function is classified into one of these barren levels based on the actual percentage of barren plateaus present in the function, estimated as described in App.~\ref{app:actual_barren_percentage_estimation}. Finally, in order to obtain a balanced distribution of functions assigned to all three barren levels, synthetic functions are generated until 9 functions are accepted for analysis in each level.

In an attempt to reduce the probability of a lucky start that may disrupt overall conclusions about the impact of the barren level on the method performances, we have restricted the choice of the ans\"atze used in each optimization experiment to a uniformly random selection of the grid points in the nominal barren areas, as long as their distance to the global optimum is larger than $0.6$ (i.e. $60\%$ of the parameter space range).

Tab.~\ref{tab:parameterization_synthetic} lists the parameters defining the synthetic optimization problems.

\begin{table}[htbp]
\centering
\scalebox{1.0}{
\begin{tabular}{|c|c|c|}
\hline
\multicolumn{1}{|c|}{\textbf{Parameter}}                             & \multicolumn{1}{l|}{\textbf{Value range}} & \multicolumn{1}{l|}{\textbf{Value(s) used}} \\ \hline
Domain                                                               & Bounded in $\mathbb{R}^2$                 & $[0, 1]\times [0, 1]$                    \\ \hline
\begin{tabular}[c]{@{}c@{}}Nominal barren \\
percentage\end{tabular}& 0\%-100\%                                 & $25\%, 50\%, 80\%$                       \\ \hline
\end{tabular}
}
\caption{Parameters defining the synthetic problems.}
    \label{tab:parameterization_synthetic}
\end{table}

\subsection{QAOA Experiments on the Max-Cut problem}
Given a weighted graph, the simplest Max-Cut problem consists of finding the optimum
partition of the graph into two sets of nodes that maximizes the total weight of the crossing edges.
 
In this section we use a Quantum Approximate Optimization Algorithm (QAOA) \cite{farhi2014quantum, farhi2016quantum} (see also \cite{zhou2020quantum} for a survey), a VQA proposed as an easier alternative than classical optimization to solve combinatorial optimization problems that are e.g. NP-hard, as is the case with the Max-Cut problem.

Specifically, in QAOA two variational parameter vectors are utilised, $\beta, \gamma \in \mathbb{R}^L$, which act on $L$ layers that make up a quantum circuit consisting of products of a mixer unitary, $U_M(\beta_l) = e^{-i \beta_l X}$, $\beta_l \in [0, \pi]$, and a problem unitary, $U_P(\gamma_l) = e^{-i \gamma_l f}$, $\gamma_l \in [0, 2\pi]$, $l = 1, \dots, L$. Here $i$ represents the imaginary unit, $X=\sum_j X_j$ denotes the sum of Pauli-$X$ gates applied to the different qubits indexed by $j \in \{1, \dots, n\}$, and $f$ is the cost function of the optimization problem to minimize, encoded by the QAOA representation. When the quantum circuit obtained by the QAOA encoding is fed with the uniform superposition (known as Hadamard operation) of the $n$-dimensional basis states of $n$ qubits, $\left|+\right\rangle=\frac{1}{\sqrt{2^n}} \sum_{z \in\{0,1\}^n}|z\rangle$, the quantum state obtained at the output is computed from the circuit unitaries as
$$
\left|\psi_{\gamma, \beta}\right\rangle=U_M( \beta_L) U_P(\gamma_L) \ldots U_M( \beta_1) U_P(\gamma_1)|+\rangle.
$$
An empirical estimate $\widehat{f}$ of the expectation value of the cost function $\langle \psi_{\gamma, \beta} | f | \psi_{\gamma, \beta}\rangle$ is then measured on a specified number of shots run on the quantum circuit, and the goal is to classically optimize the variational parameters that minimize the quantum-approximated expected cost.

The QAOA circuit approach to solve the Max-Cut problem \cite{farhi2014quantum, farhi2016quantum}  defines the number of qubits feeding the quantum circuit as the number of nodes in the weighted graph instance of the problem, and encodes the original maximization problem as the minimization of the QAOA-encoded cost function,
\begin{align*}
    f_{} = \sum_{(i, j) \in E} w_{i j} Z_i \otimes Z_j,
\end{align*}
where $w_{ij}$ is the weight of the edge connecting nodes $i$ and $j$
and $Z_i$ is the Pauli-$Z$ gate operating on qubit $i$, computed as $I \otimes \cdots I \otimes Z \otimes I \cdots \otimes I$, where $Z$ is placed at the $i$-th position and $I$ is the identity gate \cite{Hadfield2021representation}. 

The experiments were carried out on the QAOA approximation of the Max-Cut problem for four different connectedness levels: $25\%, 50\%, 75\%, 100\%$ on an 8-node graph with edge weights chosen uniformly at random between 0 and 1. The number of shots for the evaluation of the quantum circuit is 512.
Tab. \ref{tab:parameterization_maxcut} lists the parameters defining the Max-Cut optimization problems.
\begin{table}[htbp]
\centering
\scalebox{1.0}{
\begin{tabular}{|c|c|c|}
\hline
 \textbf{Parameter}&\textbf{\begin{tabular}[c]{@{}c@{}}Value \\
range\end{tabular}}        & \textbf{Value(s) used}  \\ \hline
 No. nodes (qubits) & $\geq 2$ & 8                    \\ \hline
 Edge weights       & $\mathbb{R}$ & Unif. dist. in $(0, 1)$ \\ \hline
 \begin{tabular}[c]{@{}c@{}}Graph \\
connectedness\end{tabular} & $> 0$     & 25\%, 50\%, 75\%, 100\%\\ \hline
 \begin{tabular}[c]{@{}c@{}}$L$: No. layers in \\
 QAOA circuit\end{tabular} & $\geq 1$ & \begin{tabular}[c]{@{}c@{}}1 (=> 2 optimization \\
 parameters)\end{tabular}\\ \hline
\end{tabular}
}
\caption{Parameters defining the Max-Cut problems.}
    \label{tab:parameterization_maxcut}
\end{table}

% \begin{figure}[htbp]
%         \centering
%     \includegraphics[scale=0.3]{figures/Results-MAXCUT_graph_weighted_connectedness=1.0_weightvalues.png}
%     \caption{Instance of a fully-connected 8-node weighted graph used in the analyzed Max-Cut problem where edge weights are selected from a uniform distribution in $[0, 1]$.}
%     \label{fig:maxcut_graph}
% \end{figure}

%\textcolor{blue}{Unless indicated otherwise, }

\subsection{Experimental results}\label{sec:results}
This section presents the experimental results obtained for the two problems presented in the previous sections: the synthetic case and the Max-Cut problem.

Each analyzed optimization method is run on each realized function until 50 optimization steps are reached, using $N=10$ FV particles for the FV-based methods. This implies that the ``SA-1'' method is run for 500 ($= 50 \times 10$) optimization steps. Ten replications of the optimization process (each replication runs on a different seed and on different ans\"atze) are run for each function realization.

In order to make results comparable across the two problems and across the different experiments, we define a performance metric that is independent of the global minimum value attained by the function and of the function range. The performance metric is thus defined as the minimum value found by the analyzed method on the function after standardization to the interval $[0, 1]$, as: $(f - f_{min}) / (f_{max} - f_{min})$, where $f_{min}$ and $f_{max}$ are respectively the true global minimum and maximum of the realized function, computed on a regular grid of $100 \times 100$ in the synthetic case, and of $50 \times 50$ in the Max-Cut case \footnote{Even though a natural performance metric would be the relative error between the minimum value found by the analyzed method and the true global minimum, it is a metric that is highly dependent on the global minimum value and the function range. For instance, a function with global minimum equal to 1.0 and global maximum equal to 1.5 will most likely allow the optimization methods reach a smaller relative error than the one obtained on a function with the same global minimum but a global maximum equal to 10. In fact, the maximum possible error for the latter function is 900\%, whereas it is just 50\% for the former function. Therefore, the relative error metric is not comparable across these two functions.}. The best possible performance is thus 0 and the worst is 1 \footnote{In certain cases, the best performance could be slightly negative, due to the fact that the true global minimum is computed on a regular grid, whereas the estimated global minimum is not restricted to the grid.}.

The performance metric is analyzed in terms of the amount of barren plateaus present in the optimization function, estimated using the procedure described in App.~\ref{app:actual_barren_percentage_estimation} as a percentage of the domain size.
%Note that this procedure allows us to have a measure of barrenness for the Max-Cut instances.
Note that, despite setting the desired barren level in the synthetic function, the actual barren percentage of the realized function may be substantially different from the nominal value, due to the smoothing step of the function synthesis procedure
\footnote{See the end of App.~\ref{app:synthetic_function_generation} for a restriction imposed on the relationship between the estimated barren percentage and the nominal one to accept the generated function as a valid synthetic function to be used in experiments.}. For the final analysis of performance, the estimated barren percentages are used to group the generated functions into one of the following three barren levels: LOW [0\%, 33\%], MEDIUM (33\%, 66\%], and
HIGH (66\%, 100\%), which are analyzed separately.

Two types of performance plots on the selected metric are created: violin plots and convergence speed plots. Fig.~\ref{fig:results_violin_by_barren_level} presents violin plots of the standardized minimum value found by each method across the three different barren levels, while Fig.~\ref{fig:results_violin_overall} presents the same results summarized over all experiments regardless of barren level.
Fig.~\ref{fig:results_speed} can be used to compare the speeds with which each algorithm reaches its estimated minimum. The plots in the latter figure are generated as follows: first, to average out the randomness coming from the ten replications run on each function instance, a summary statistic of the minimum function value found at each optimization step is computed over all such replications; then, the distribution of this summary statistic is computed over all functions assigned to each barren level. For robustness, the chosen summary statistic is the median.

We observe that FV-based methods tend to outperform vanilla SA, both in terms of the minimum function value found and its variability (Figs.~\ref{fig:results_violin_by_barren_level} and \ref{fig:results_violin_overall}), and in terms of convergence speed to such minimum (Fig.~\ref{fig:results_speed}), specially when compared against the ``SA-1'' method. We note that all methods perform better in the Max-Cut problem than in the synthetic case, both in terms of smaller minimum value found and smaller variability of the minimum.
%as the distribution of the standardized minimum value found is closer to 0 for all of them.
The smaller variability is most likely due to high similarities among Max-Cut instances, whereas the synthetic functions may have very different landscapes (see examples in Figs.~\ref{fig:results_synthetic_instances} and \ref{fig:results_maxcut_instances} of App.~\ref{app:actual_barren_percentage_estimation}). In both cases, methods tend to perform from better to worse in the order they are laid out in the plots, from ``FV-100\%'' to ``SA-1''. The FV advantage over SA tends to be larger for higher barrenness, as measured by the median performance. This is specially true in the synthetic case, where the advantage in median value between ``FV-100\%'' and ``SA-1'' goes from $\sim 0.2$ for the LOW barren level, to $\sim 0.3$ for the MEDIUM barren level, to $\sim 0.5$ for the HIGH barren level, all of which values are to be compared with the maximum possible advantage of $1$.
We also observe that FV with $100\%$ exploration rate tends to perform slightly better than FV with $50\%$ exploration rate.

In Fig.~\ref{fig:results_speed} we observe a noticeable faster convergence of the FV methods to their respective found minimum, specially in the synthetic case, where in median the two green curves are clearly below the two red curves. This is specially the case in the HIGH barren level group where both FV methods reach almost their respective optimum half way, i.e., at iteration $\sim 25$, in median. On the other hand, in the Max-Cut problem, the speed observed in FV-based methods is very similar to the speed of the ``SA-N'' method, although convergence is slightly faster in the HIGH barren level group where, in median, they reach the global minimum at iteration $\sim 20$, compared to iteration $50$ for ``SA-N''. In all cases, ``SA-1'' stays in median $\sim 0.2$ standardized units away from the global minimum.

Fig.~\ref{fig:results_maxcut_trajectories} was produced as an illustration of the trajectories of the FV-based methods and their comparison with the SA trajectories.
%The ``SA-N'' trajectories (red) and the ``FV-100\%'' trajectories (green) are overlaid on the contour plot of one Max-Cut instance. 
We see how each ansatz of the SA algorithm stays exploring around the same area, trapped by barren plateaus, whereas FV jumps into other possibly more interesting regions and finally is able to find the global minimum.

%%%%%%%%%%% VIOLIN PLOTS BY BARREN LEVEL
\begin{figure}[htbp]
    \centering
    \begin{minipage}{\columnwidth}
        \captionsetup{justification=centering}
        \subfloat[\footnotesize Synthetic case]{\includegraphics[width=\linewidth]{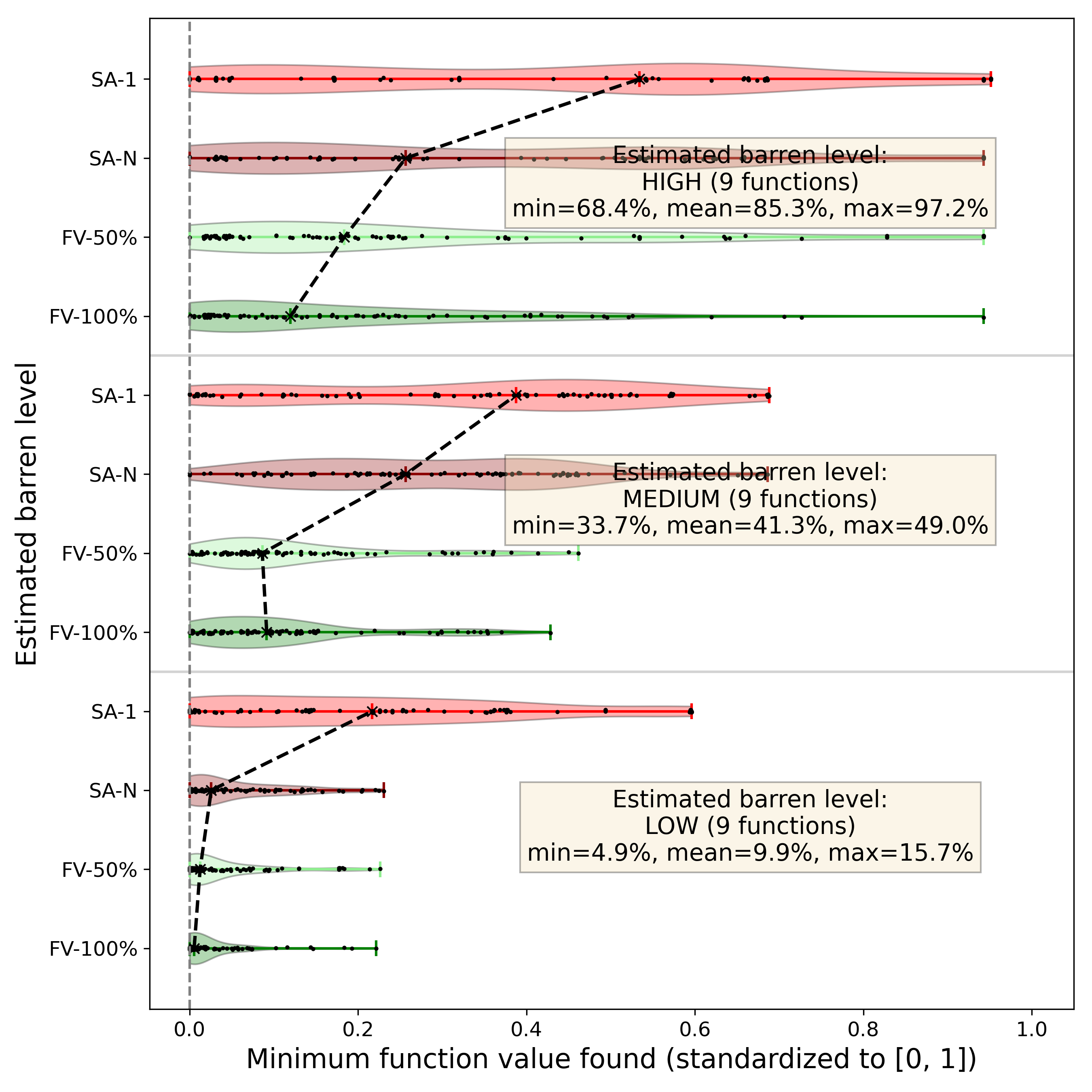}}
    \end{minipage}
    \begin{minipage}{\columnwidth}
        \captionsetup{justification=centering}
        \subfloat[\footnotesize Max-Cut problem]{\includegraphics[width=\linewidth]{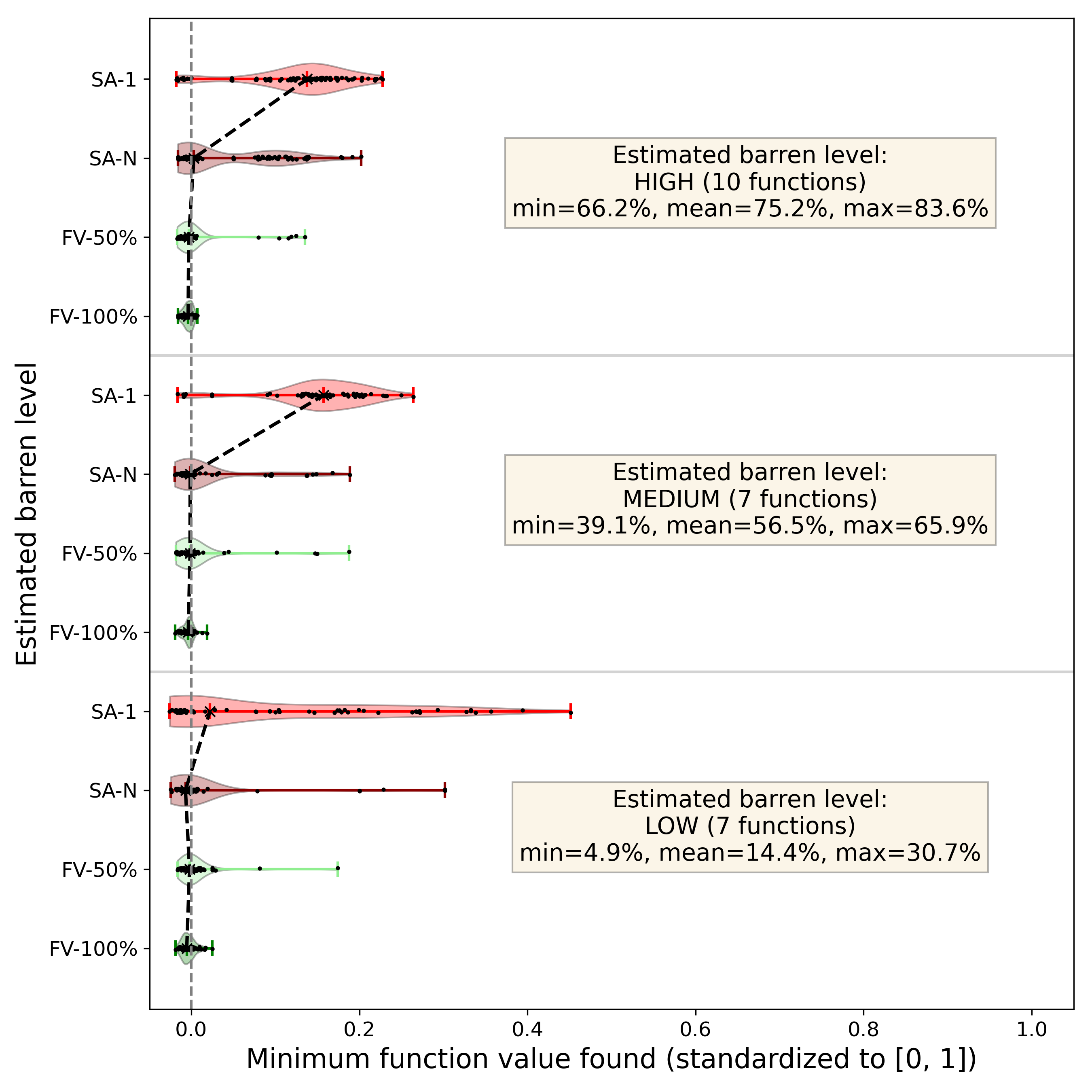}}
    \end{minipage}
    \caption{Violin plots by estimated barren level of the minimum value found by each method for the function instances standardized to the $[0, 1]$ interval. (a) Synthetic case with 27 instances, (b) Max-Cut problem with 24 instances. The range of estimated barren percentages for the functions in each barren level is indicated in the insets. The black line in each group connects the medians. Lower values (more to the left) are better. Negative estimated function minima may arise due to the approximate computation of the true global minimum on a regular grid, whereas the minimum obtained by the optimization methods is not restricted to the grid.}
    \label{fig:results_violin_by_barren_level}
\end{figure}
%%%%%%%%%%% VIOLIN PLOTS BY BARREN LEVEL

%%%%%%%%%%% VIOLIN PLOTS OVERALL
\begin{figure}[htbp]
    \centering
    \begin{minipage}{\columnwidth}
        \centering
        \captionsetup{justification=centering}
        \subfloat[\footnotesize Synthetic case]{\includegraphics[height=1.15\linewidth]{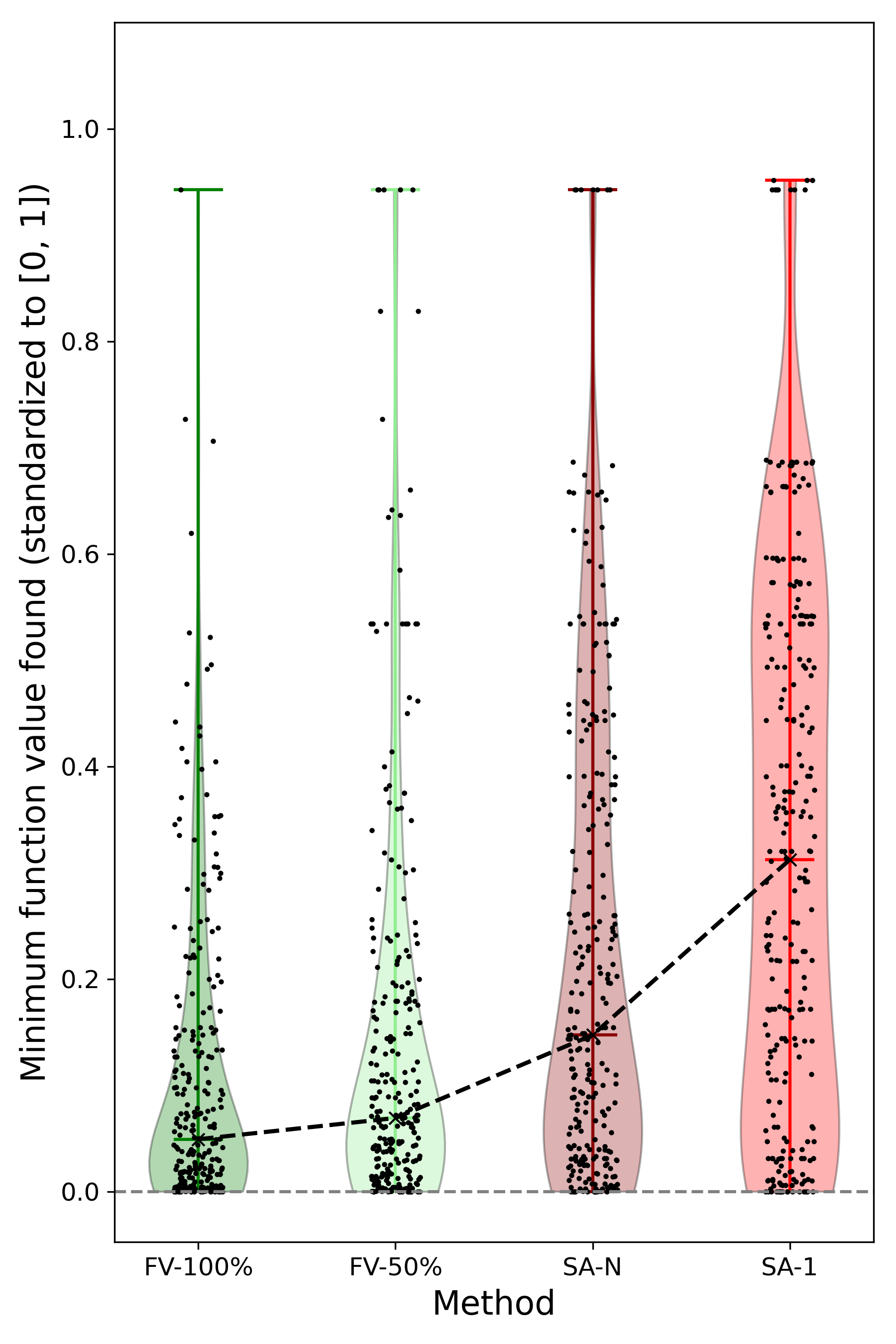}}
    \end{minipage}
    \begin{minipage}{\columnwidth}
        \centering
        \captionsetup{justification=centering}
        \subfloat[\footnotesize Max-Cut problem]{\includegraphics[height=1.15\linewidth]{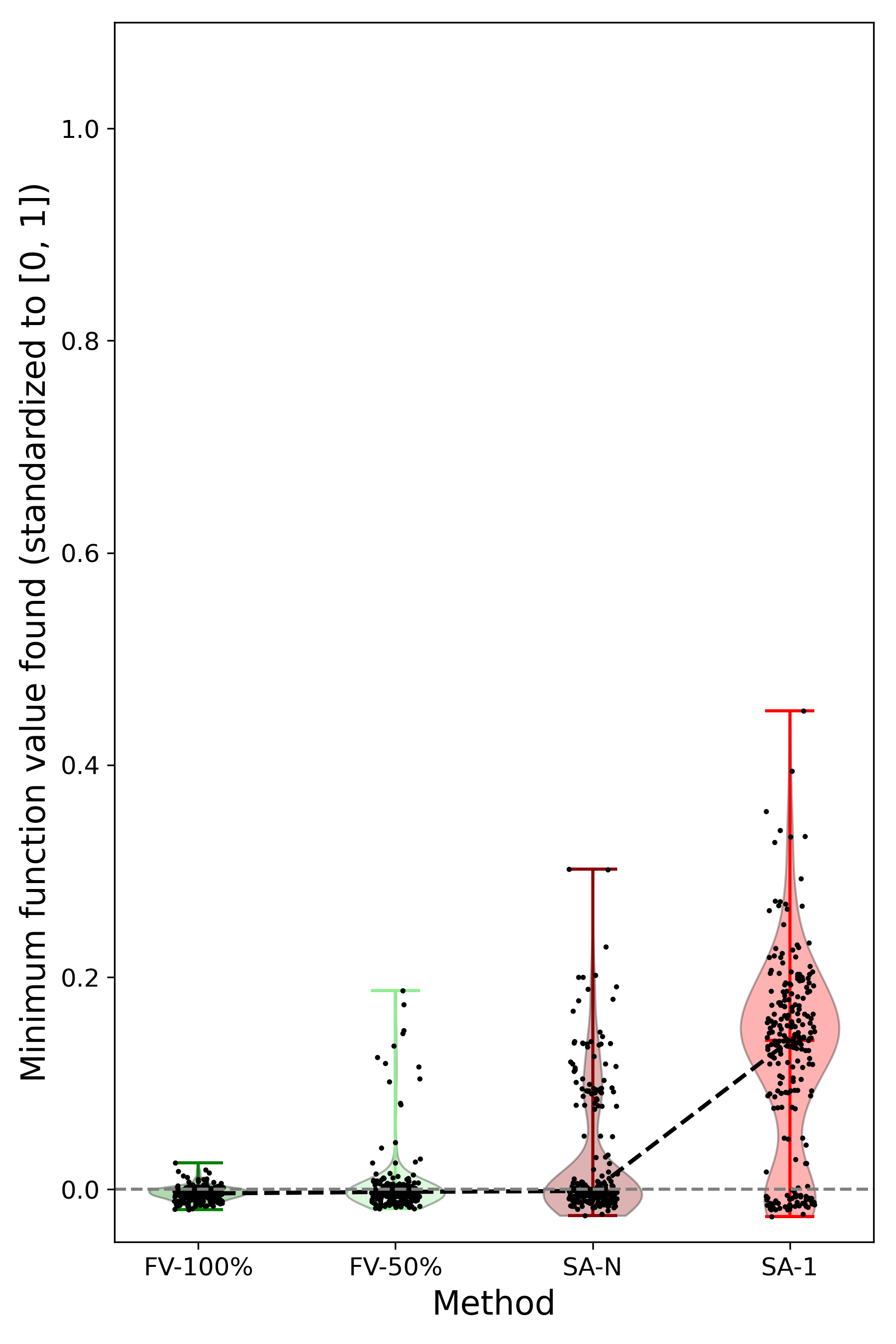}}
    \end{minipage}
    \caption{Violin plots over all analyzed instances of the minimum value found by each method for the function instances standardized to the $[0, 1]$ interval. (a) Synthetic case with 27 instances, (b) Max-Cut problem with 24 instances. Lower position of the violins is better. Negative estimated function minima may arise due to the approximate computation of the true global minimum on a regular grid, whereas the minimum obtained by the optimization methods is not restricted to the grid.}
    \label{fig:results_violin_overall}
\end{figure}
%%%%%%%%%%% VIOLIN PLOTS OVERALL

%%%%%%%%%%% SPEED
\begin{figure}[htbp]
    \centering
    \begin{minipage}{\columnwidth}
        \centering
        \captionsetup{justification=centering}
        \subfloat[\footnotesize Synthetic case]{\includegraphics[height=1.15\linewidth]{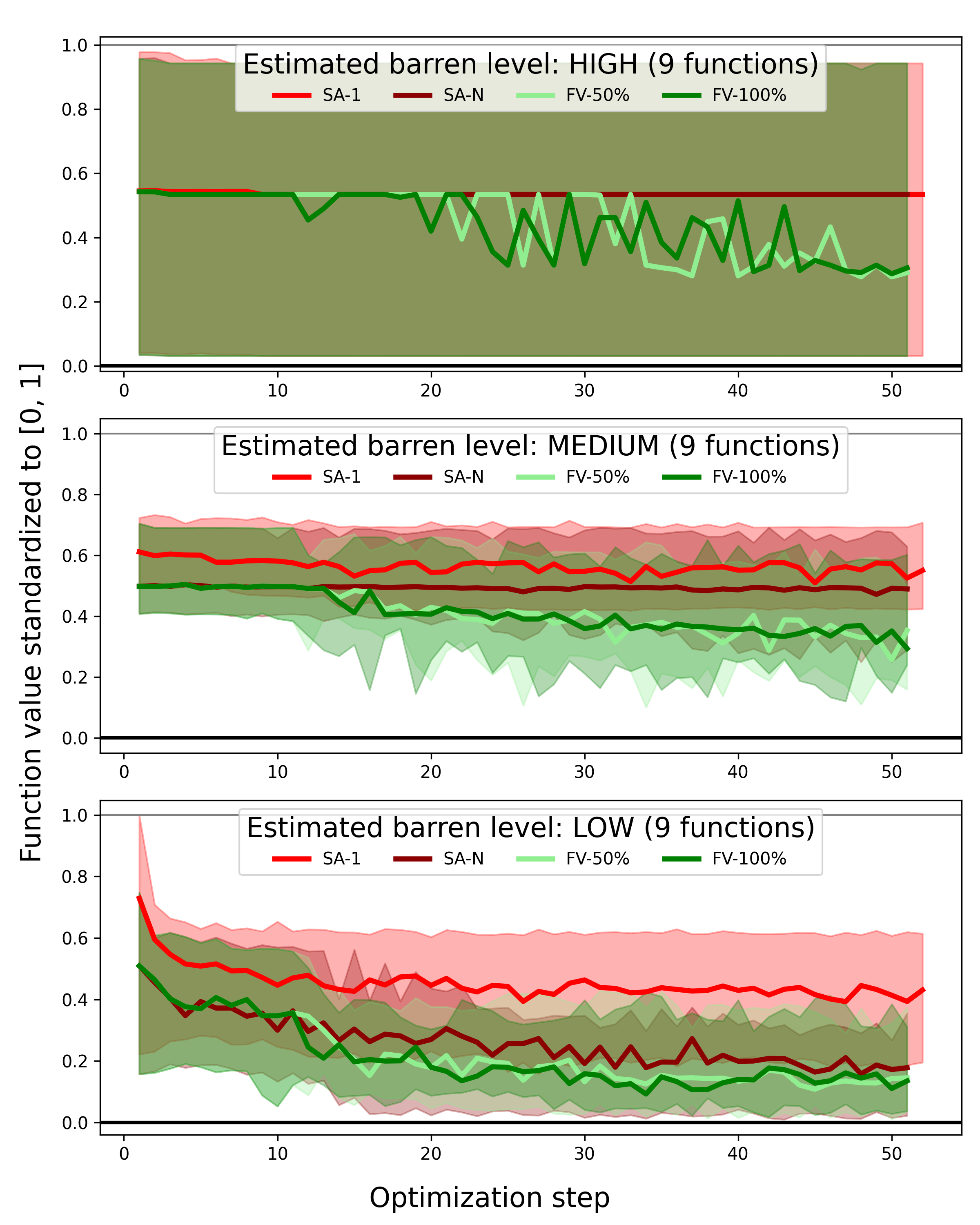}}
    \end{minipage}
    \begin{minipage}{\columnwidth}
        \centering
        \captionsetup{justification=centering}
        \subfloat[\footnotesize Max-Cut problem]{\includegraphics[height=1.15\linewidth]{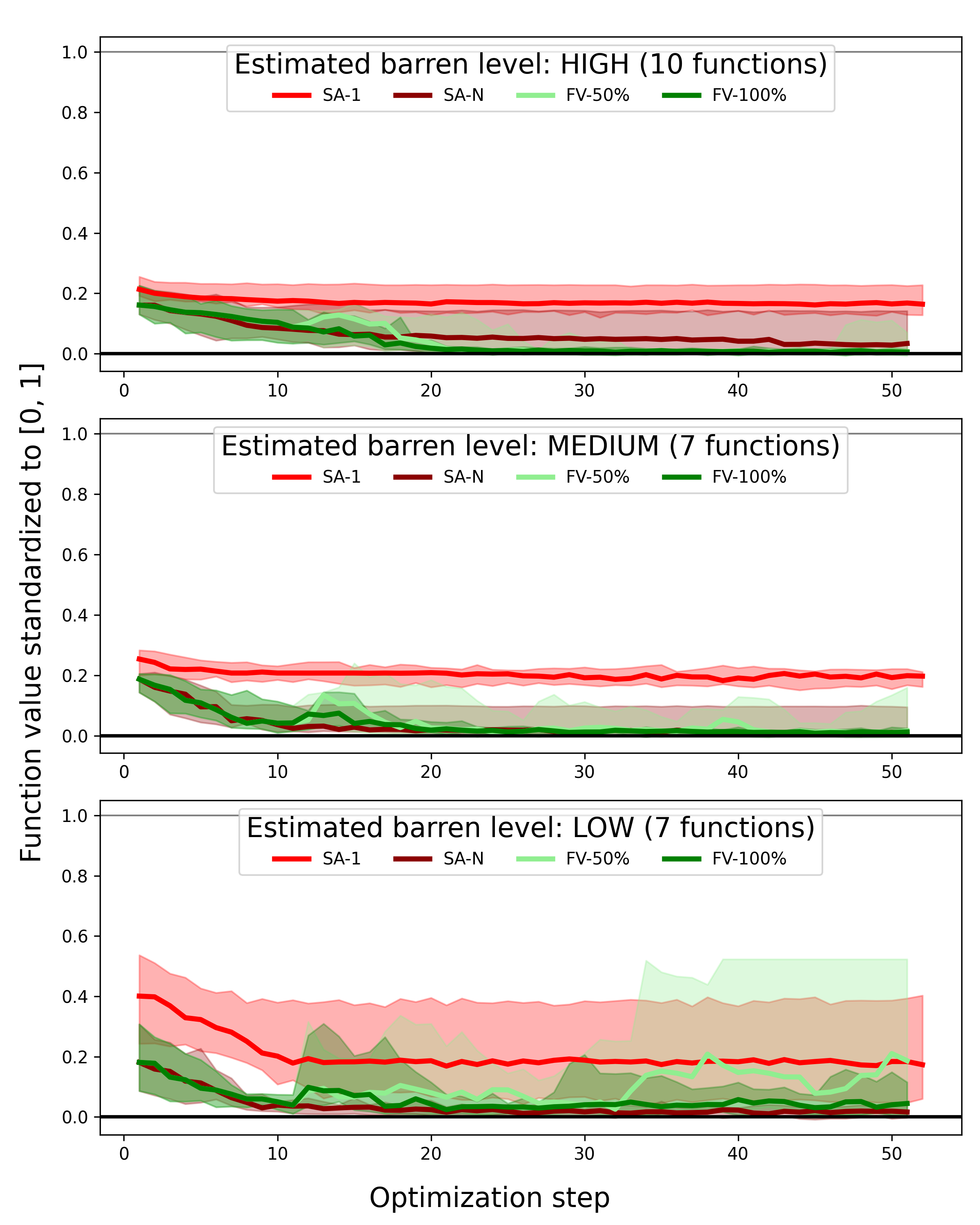}}
    \end{minipage}
    \caption{Speed by estimated barren level with which each method reaches the estimated function minimum on the function instances standardized to the $[0, 1]$ interval. (a) Synthetic case with 27 instances, (b) Max-Cut problem with 24 instances. Each panel shows the distribution of the median, computed as described in Section~\ref{sec:results}.
    The thick coloured curves represent the median of the distribution, and the bands span the distribution's min-max range. Lower positions of the curves are better. %Although 500 optimization steps are used in the ``SA-1'' method, only the first 50 are shown.
    }
    \label{fig:results_speed}
\end{figure}
%%%%%%%%%%% SPEED

%%%%%%%%%%% TRAJECTORIES
\begin{figure}[htb]
    \centering
    \includegraphics[width=1\linewidth,height=1\linewidth]{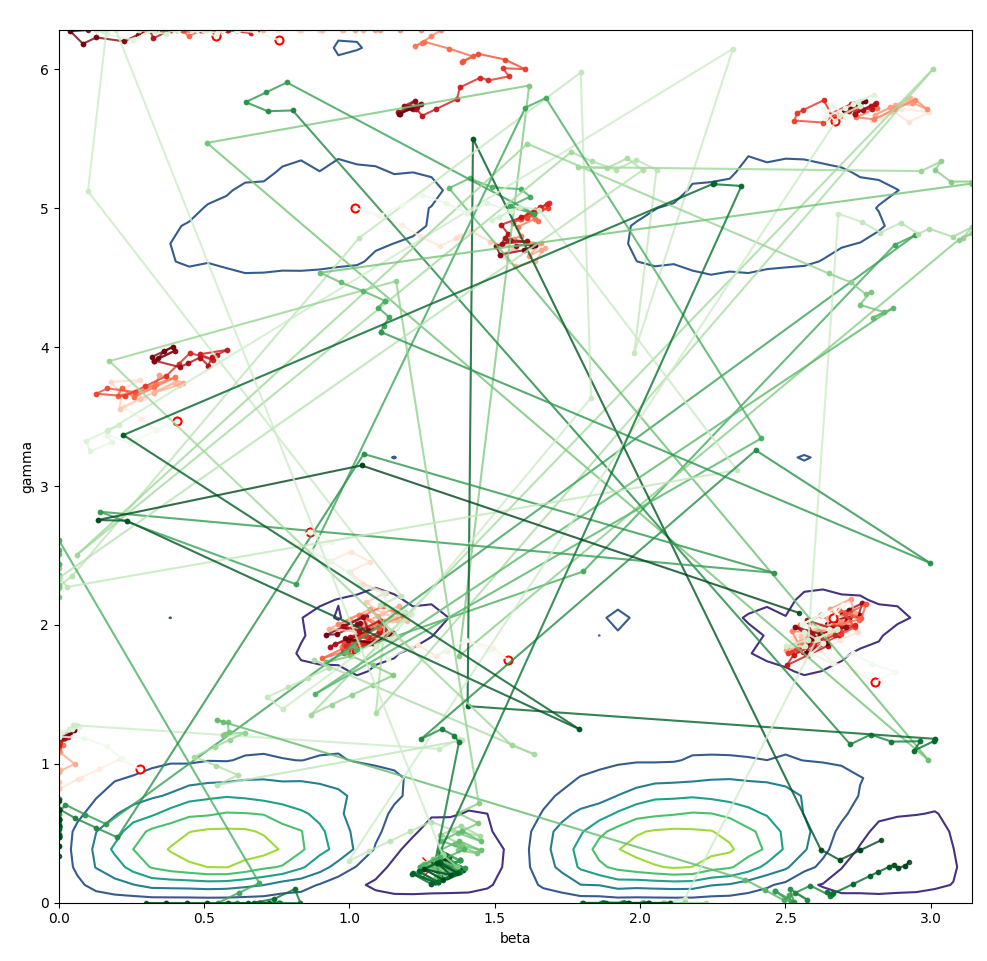}
    \caption{Trajectories of the 10 ans\"atze used in one replication of the ``SA-N'' (red) and ``FV-100\%'' (green) algorithms overlaid on the contour plot of a Max-Cut instance. The color of each trajectory intensifies as the optimization progresses. The global optimum, which occurs at $\sim (1.3, 0.25)$ at the bottom of the domain, is found by ``FV-100\%'' but not by ``SA-N''.}
    \label{fig:results_maxcut_trajectories}
\end{figure}
%%%%%%%%%%% TRAJECTORIES

\section{Summary and conclusions}\label{sec:sumconc}
We have explored a novel approach to enhancing variational quantum algorithms, while taking inspiration from the Fleming-Viot stochastic process, traditionally applied in modeling biological evolution. Our primary objective was to address the challenge of barren plateaus in variational quantum algorithms, i.e. the regions in the parameter space where gradients are very small or excessively noisy, impeding effective optimization.

We proposed a parallel implementation of the classical learning step in variational quantum algorithms, utilizing the concept of multiple searches (termed particles) that are terminated upon encountering barren plateaus and then regenerated
(using either an exploitation or an exploration strategy) into potentially more fruitful areas of the parameter space. This strategy tends to take the search for the global optimum away from barren plateaus, thus improving the overall efficiency of variational quantum algorithms.

Besides theoretically proving that the time to reach the global minimum is shorter than plain simulated annealing, our methodology was tested through numerical experiments on synthetic problems and instances of the Max-Cut problem on small graphs. The results demonstrated that our approach tends to outperform  simulated annealing, both in terms of reaching closer to the true global minimum and in terms of a smaller variability of the optimization process result. While the advantage of reaching closer to the true global minimum was more strongly observed in the synthetic experiments, the smaller variability was more strongly observed in the Max-Cut problem.
We also note that the synthetic functions test bench allowed us to analyze the performance of the proposed Fleming-Viot method as a function of the barren areas percentage, corroborating our analytical results showing that the advantage over simulated annealing becomes larger as the barren level increases.
Experiments also suggest that the exploration strategy for particle regeneration is more important than the exploitation strategy in terms of algorithm performance. We believe that this approach not only enhances the optimization process but also offers a fresh perspective on tackling the inherent challenges in quantum computing optimizations.

Our research contributes to the ongoing efforts in quantum computing to realize the potential of variational quantum algorithms for near-term advantage, especially in the face of obstacles like barren plateaus. We believe that the insights gained from our work will help open new research directions in near-term quantum algorithms, particularly in developing more robust and efficient classical optimization strategies for variational quantum algorithms.

\paragraph*{Acknowledgements}
We would like to thank Matthieu Jonckheere, from CNRS-LAAS in Toulouse, France, for the insights received during the development of the algorithm.
Jakub would like to acknowledge many discussions with Andrew J. Parkes up to 2012 concerning landscapes and algorithm analysis. This work has been partially supported by project MIS 5154714 of the National Recovery and Resilience Plan Greece 2.0 funded by the European Union under the NextGenerationEU Program. 

\paragraph*{Disclaimer} This paper was prepared for information purposes,
and is not a product of HSBC Bank Plc. or its affiliates.
Neither HSBC Bank Plc. nor any of its affiliates make
any explicit or implied representation or warranty and
none of them accept any liability in connection with
this paper, including, but limited to, the completeness,
accuracy, reliability of information contained herein and
the potential legal, compliance, tax or accounting effects
thereof. Copyright HSBC Group 2023.

\newpage

%============== APPENDIX
\newpage
\appendix

\section{Fleming-Viot optimization algorithm}
\label{app:algorithm}

\begin{algorithm}[H]%[!htb] % We need to use `[H]` in order to avoid the algorithm to be floating and thus avoid the error with the `\endcsname` command (Ref: https://tex.stackexchange.com/questions/231191/algorithm-in-revtex4-1)
\caption{FV parallel optimization algorithm}\label{alg:FV}
\begin{algorithmic}[1]  % The number in brackets requests line numbering (the number itself represents the number to assign to the first line)

\Require 
$T$: the number of classical optimization iterations,

$N$: the number of Fleming-Viot particles,
        
$B$: the burn-in time steps used to compute the reference gradient norm involved in the particle's killing (a.k.a. absorption) condition,

$W$: the size of the most recent time window on which the average gradient is computed,

$\alpha$: the relative gradient threshold for absorption of a particle,

$\varepsilon$: the exploration rate.

\Ensure 
    Estimate of the parameter vector minimizing the system energy: \State $\text{avg\_abs\_gradient} \gets 0$

\For{$i = 1$ to $B$}
    \State \Call{OptimizeAndRecordGradient}{}
    \State Update $\text{avg\_abs\_gradient}$
\EndFor

\State $\text{ref\_gradient} \gets \alpha \times \text{avg\_abs\_gradient}$

\For{$i = 1$ to $T$}
    \State \Call{OptimizeAndRecordGradient}{}
    \State \Call{CheckAbsorptionAndRegenerateParticles}{}
\EndFor

\end{algorithmic}
\end{algorithm}

\begin{algorithm}[th]
\begin{algorithmic}

%\medskip
\Procedure{OptimizeAndRecordGradient}{}
    \For{$j = 1$ to $N$}
        \State Perform optimization step for particle $j$
        \State Record observed gradient for particle $j$
    \EndFor
\EndProcedure

\medskip
\Procedure{CheckAbsorptionAndRegenerateParticles}{}
    \For{$j = 1$ to $N$}
        \If{avg gradient of particle $j$ over last $W$ steps < $\text{ref\_gradient}$}
            \State $\text{rnd} \sim \mathcal{U}[0, 1]$
            \If{$\text{rnd} < \epsilon$} \State \Call{ReinitializeParticle}{j}
            \Else \State \Call{ReactivateParticle}{j}
            \EndIf
        \EndIf
    \EndFor
\EndProcedure

\medskip
\Procedure{ReinitializeParticle}{j}
    \State Regenerate particle $j$ to a randomly chosen point in the parameter space
    \State Reset the history of recorded gradient values for particle $j$
\EndProcedure

\medskip
\Procedure{ReactivateParticle}{j}
    \State Regenerate particle $j$ to a randomly chosen particle ($r$) among non-killed particles
    \State Assign the last $W$ gradient values of particle $r$ as the last $W$ gradient values of particle $j$
    \State Reset the learning rate of particle $j$ to the initial learning rate of particle $r$ with the goal of increasing the divergence between the two particles, thus boosting exploration
\EndProcedure

\end{algorithmic}
\end{algorithm}

%\newpage

\section{Procedure to generate 2D smooth synthetic functions}
\label{app:synthetic_function_generation}
 The following procedure was used to generate the synthetic functions that serve as test bench for the proposed optimization algorithm. Define:
\begin{itemize}
    \item a 2D rectangular parameter domain on which the synthetic function will be generated (e.g. $[0, 1] \times [0, 1]$), having an area of $A$ grid points,
    \item the number of grid points to consider in each 2D dimension on which the synthetic function $f$ will be defined (e.g. 100 points),
    \item the desired nominal percentage $P$ of barren plateaus to be present in the function,
    \item the desired number $B$ of disjoint barren plateaus,
    \item the width of the border margin, i.e. the margin to leave between the border of the domain and the barren plateaus (e.g. 2 grid points),
    \item the number $M$ of sample points outside barren areas where desired function values are defined (e.g. 100),
    \item the smoothing parameter for the smoothing spline interpolation used to define the final synthetic function value $f$ at each grid point (e.g. 10),
\end{itemize}
the synthetic process consists of the following main steps:
\begin{enumerate}[label={\Alph*)}]
\item Generate barren areas covering at least the nominal percentage of barren plateaus in the domain and sample all grid points in them.
\item Sample $M$ points outside barren areas. 
\item Sample the nominal position of the global minimum.
\item Generate the synthetic function by a smooth interpolation of all sampled points.
\item Estimate the actual barren percentage of the generated function.
\item Check the generated function on quality criteria which define its acceptance or rejection for use in the optimization experiments.
\end{enumerate}

Step (A) consists of the following iterative process:
\begin{enumerate}[label={(\arabic*)}]
    \item Set the value of the function at the grid points in the border margin to a value a little larger than 1 (e.g. 2), so that the global minimum of the function is less likely to be located at the border, which is an undesired situation\footnote{\label{fnt:border_minimum_situation}The situation is undesired because it artificially helps the optimization algorithm find the global optimum. This aid is due to the fact that the parameter estimate at each optimization step is clipped to the domain. Therefore, if after a particular step the parameter estimate happens to fall outside the domain, its location might end up being close to the global optimum just because of clipping.}.
    \item Randomly sample a point among the grid points in a subset of the domain that guarantees there is no overlap between the barren plateaus and the border margin; this point represents the center of a barren plateau.
    \item Construct a square centered at the sampled point with side length equal to $\lfloor\frac{\sqrt{P A / B}}{2}\rfloor$, which represents the new proposed barren plateau.
    \item Sample a value from a uniformly distributed random variable in $[0.2, 0.8]$ (note a margin is left from both $0$ and $1$) and assign it to each grid point in the proposed barren plateau.
    \item Check overlap between the newly constructed square and other barren plateaus previously constructed. If there is overlap, the regions are merged into a single barren area. The function values assigned to each point in each of the regions before merging are maintained.
    \item Compute the current barren plateau percentage as the percent of grid points in the barren areas constructed so far and the total number of grid points in the domain. If this percentage is smaller than the nominal percentage of barren plateau, go to step (2).
\end{enumerate}

Step (B) consists of uniformly sampling $M$ grid points in non-barren areas and assigning uniformly distributed random values in $[0.2, 0.8]$ as function values, i.e. the same distribution used for the function values in barren plateaus.

Step (C) consists of uniformly sampling one grid point in non-barren areas and assigning zero as the function value. This is considered the nominal position of the global minimum.

Step (D) consists of fitting a 2D smoothing cubic spline (using the selected smoothing parameter) to all the sampled function values generated above: (i) the values assigned to the grid points in barren areas, (ii) the value assigned to the grid points at the four borders of the domain, (iii) the $M$ function values in non-barren areas, (iv) the nominal global minimum. The value of the smoothing spline fit in each grid point $(x_i, y_i)$ of the domain is assigned as the synthetic function value $f(x_i, y_i)$. Note that, after this smoothing interpolation step, both the position and value of the global minimum may not be equal to their selected nominal values, and the same is true for the size and location of the nominal barren areas.

Step (E) consists of estimating the actual percentage of barren plateaus present in the synthesized function using the procedure described in App.~\ref{app:actual_barren_percentage_estimation}. The actual barren percentage is likely to be substantially smaller than the nominal barren percentage due to the final smoothing interpolation step.

Step (F) consists of running the following quality checks on the generated function $f$, whose goal is, on one hand (checks QC (1) and QC (2)), to keep function characteristics other than the percentage of barren plateaus as homogeneous as possible --as we want to concentrate on the analysis of the impact of barren plateaus percentage on the method performances-- and, on the other hand (check QC (3)), to obtain an analysis in terms of barren levels that is as clean as possible.

QC (1): The global minimum does not fall at a point in the border of the domain\footnote{Recall footnote~\ref{fnt:border_minimum_situation}.}.

QC (2): The function values are between -50 and 50, in order to avoid functions that vary too wildly.

QC (3): For nominal barren percentages $B$ smaller than the left-end of the HIGH barren level (typically 66\%), the estimated barren percentage is at most equal to its nominal value. The goal is two fold: (i) use clearly distinct function instances in terms of barren levels for the analysis of the method performances\footnote{To illustrate what we mean by ``clearly distinct function instances'': suppose 66\% is the boundary between the MEDIUM barren level and the HIGH barren level. In that case, we would like to avoid including functions in the analysis that have, e.g. 65\% of estimated barren percentage (in which case the function will belong to the MEDIUM barren level group), and e.g. 67\% of estimated barren percentage (in which case the function will be belong to the HIGH barren level group). If we accept those cases, we would include functions in the analysis that have very similar barren percentage but are classified into two different barren level groups, potentially generating undesired distortions in the results.}; (ii) use function instances on which we place enough confidence on the accuracy of the estimated barren percentage because, as mentioned above, the actual barren percentage is expected to be smaller than the nominal barren percentage, due to smoothing\footnote{This, however, is less the case for large nominal barren percentages such as 80\%, where realized functions are likely to have more barren areas than expected.}.

If any of the above quality checks fails, the generated synthetic function is discarded and the whole synthesis process is run again, until the generated synthetic function is considered valid in terms of these quality standards.

%\newpage

\begin{figure}[tb]%[!htb]
    \centering
    \begin{minipage}{\columnwidth}
        \centering
        \subfloat[\footnotesize Synthetic function instance with LOW barren level (nominal barrenness: 25\%, estimated: 11\%)]{\includegraphics[width=1\linewidth,height=0.5\linewidth]{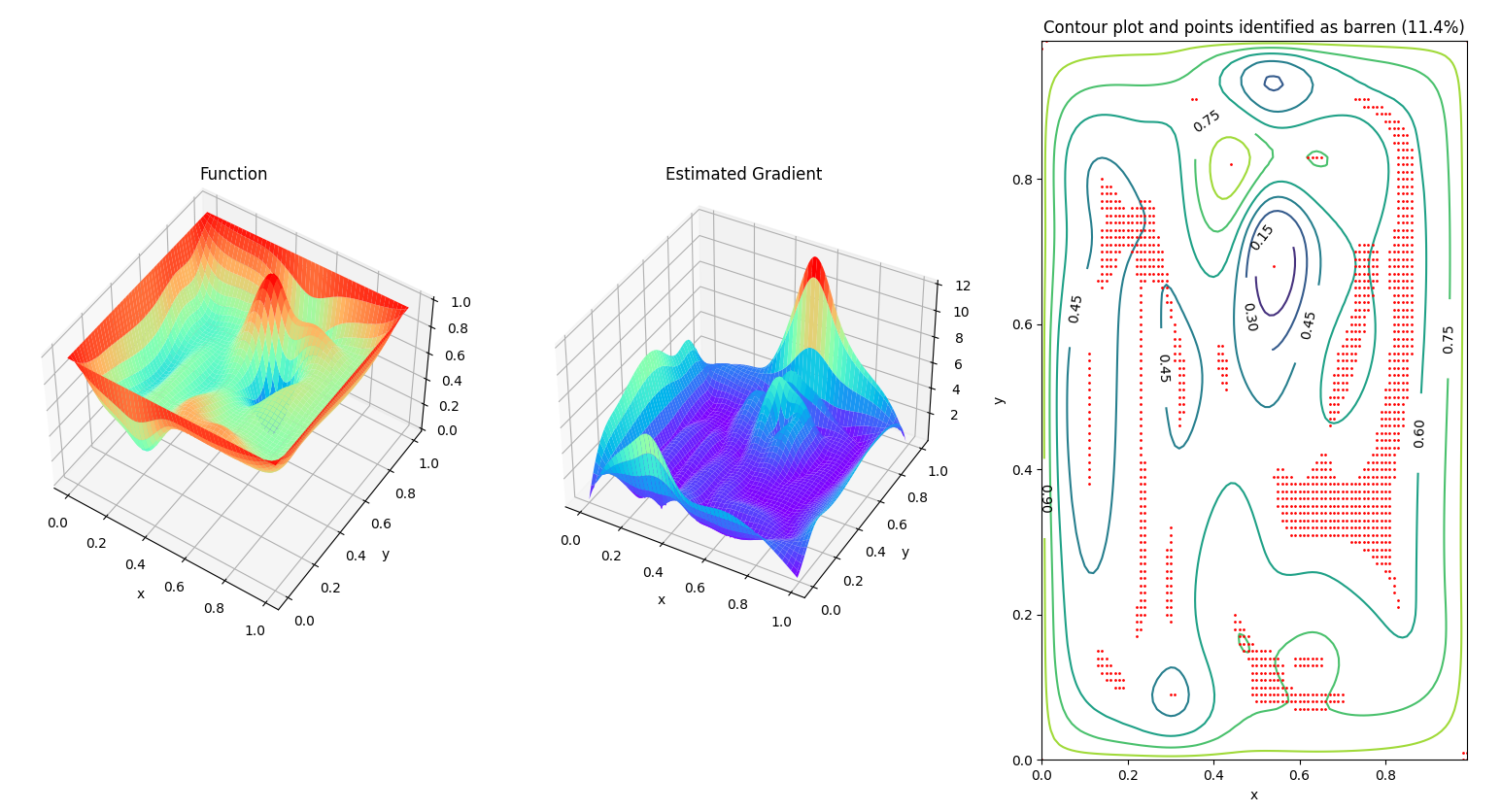}}
    \end{minipage}
    \begin{minipage}{\columnwidth}
        \centering
        \subfloat[\footnotesize Synthetic function instance with MEDIUM barren level (nominal barrenness: 50\%, estimated: 41\%)]{\includegraphics[width=1\linewidth,height=0.5\linewidth]{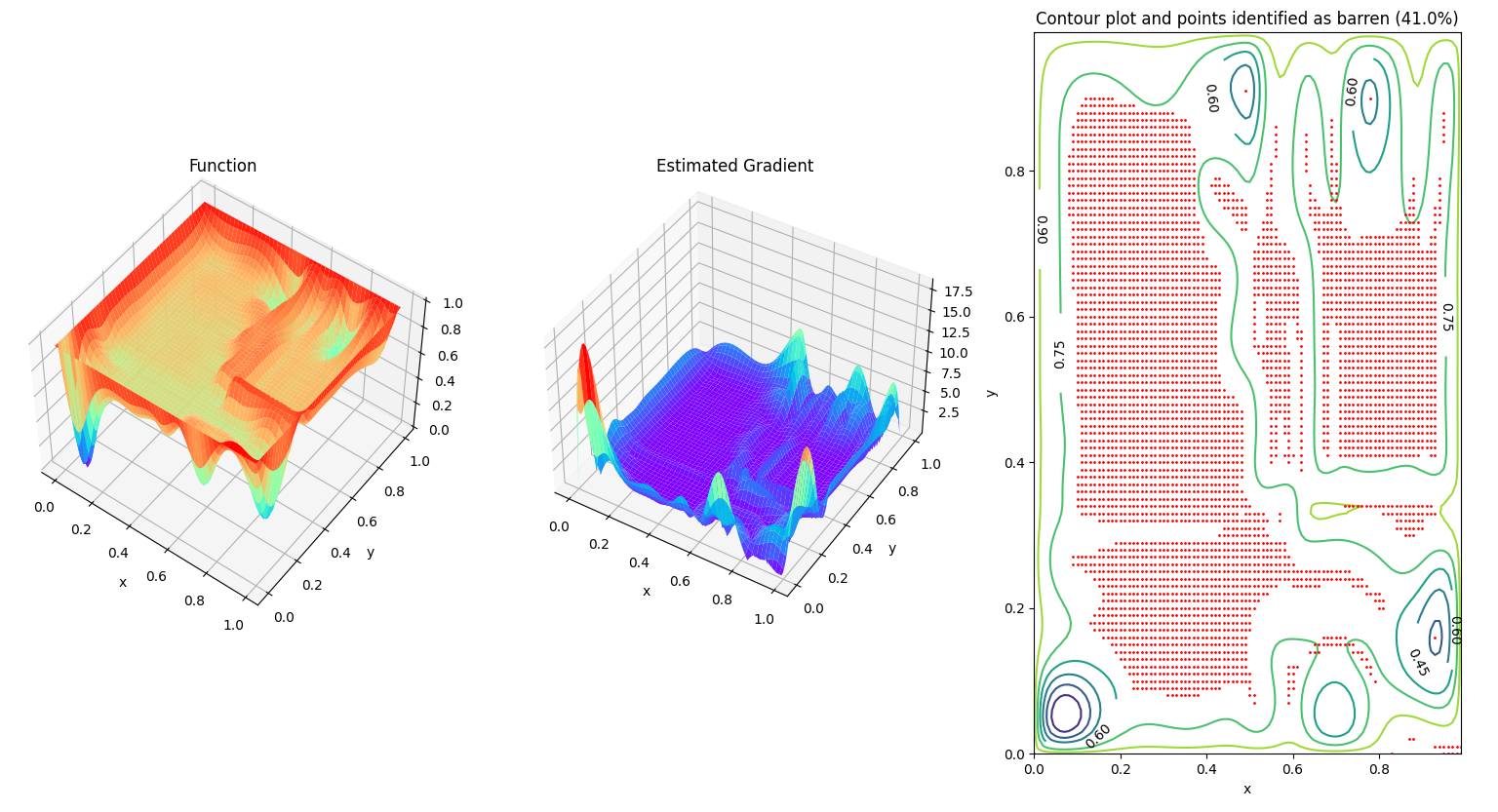}}
    \end{minipage}
    \begin{minipage}{\columnwidth}
        \centering
        \subfloat[\footnotesize Synthetic function instance with HIGH barren level (nominal barrenness: 80\%, estimated: 70\%)]{\includegraphics[width=1\linewidth,height=0.5\linewidth]{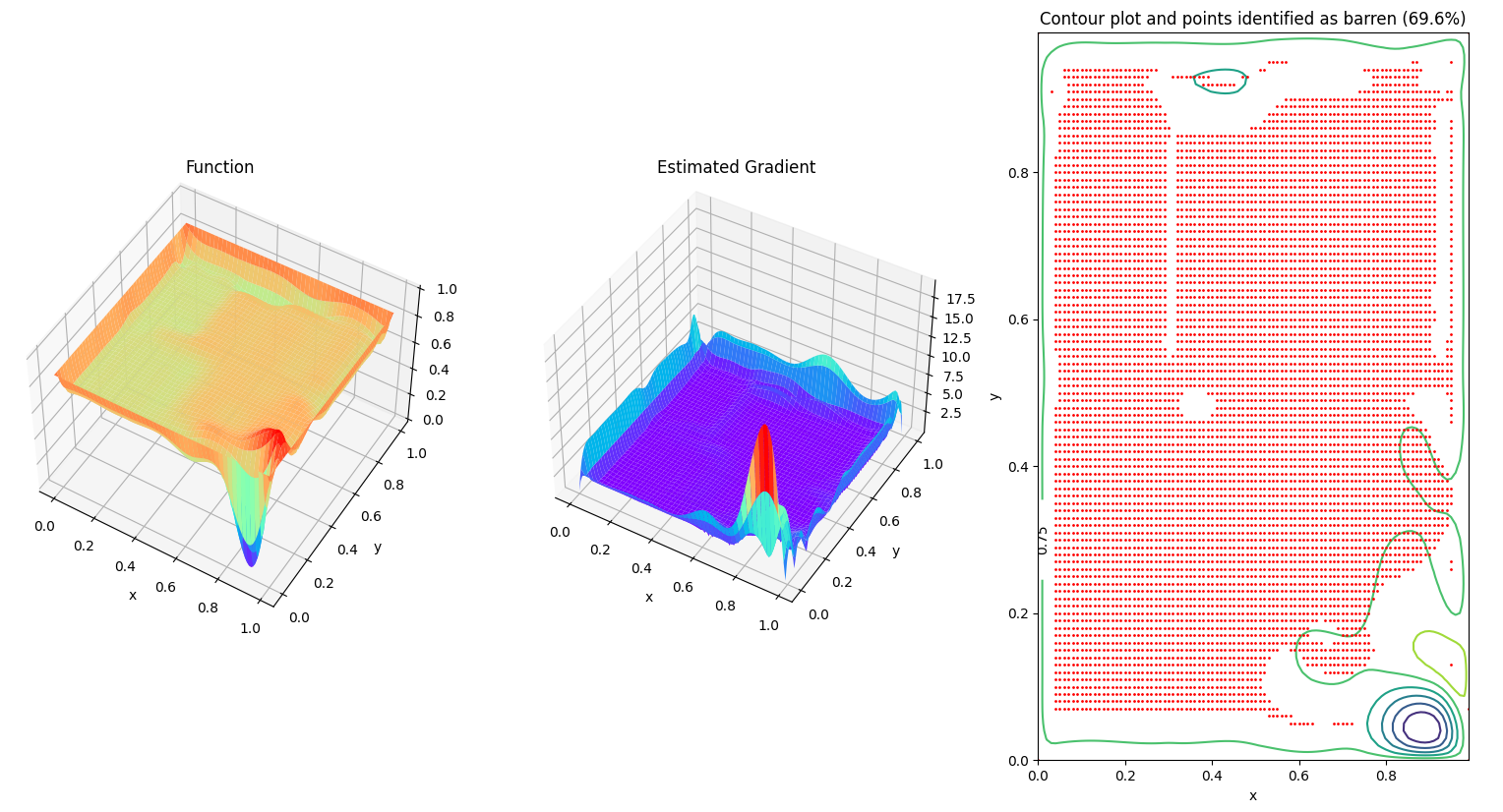}}
    \end{minipage}
    \caption{Examples of synthetic function instances assigned to each of the three possible barren levels: (a) LOW, (b) MEDIUM, (c) HIGH. Left: function landscape; center: gradient norm; right: contour plot of function with overlaid red points indicating areas identified as barren. All functions are evaluated on a $100 \times 100$ regular grid.}
    \label{fig:results_synthetic_instances}
\end{figure}

\begin{figure}[tb]%[!htb]
    \centering
    \begin{minipage}{\columnwidth}
        \centering
        \subfloat[\footnotesize Max-Cut instance with LOW barren level (graph connectedness: 25\%, estimated barrenness: 11\%)]{\includegraphics[width=1\linewidth,height=0.5\linewidth]{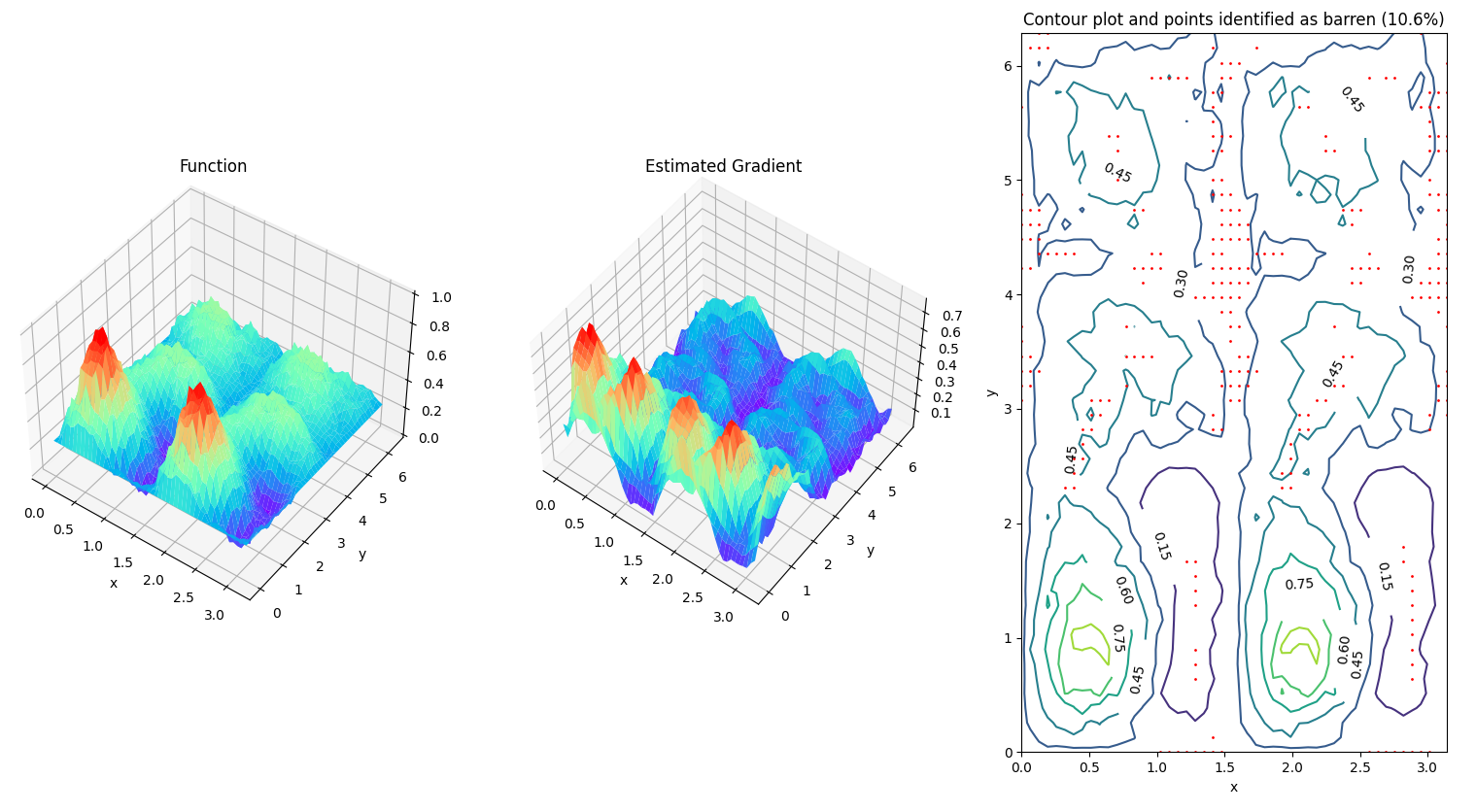}}
    \end{minipage}
    \begin{minipage}{\columnwidth}
        \centering
        \subfloat[\footnotesize Max-Cut instance with MEDIUM barren level (graph connectedness: 50\%, estimated barrenness: 49\%)]{\includegraphics[width=1\linewidth,height=0.5\linewidth]{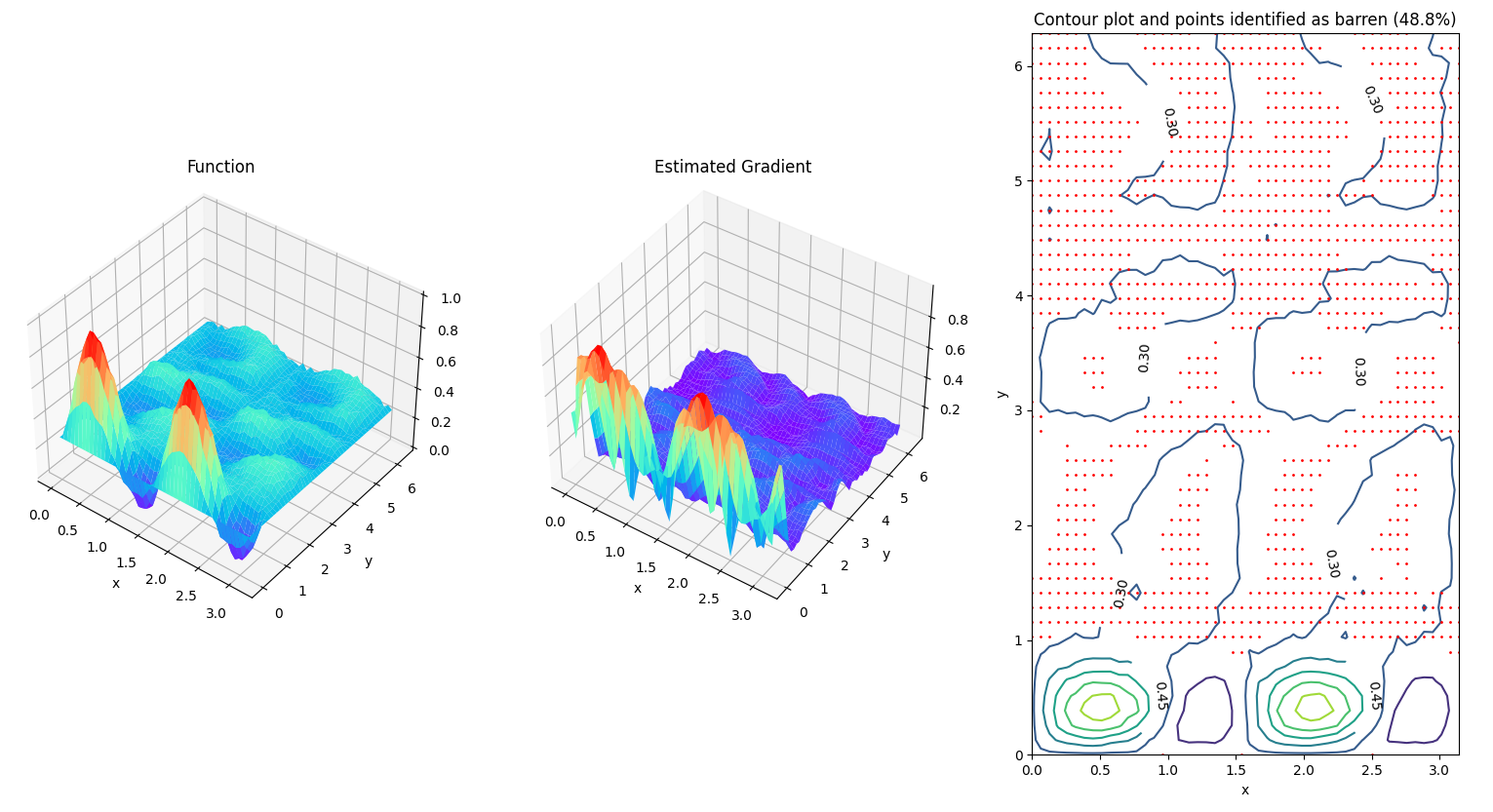}}
    \end{minipage}
    \begin{minipage}{\columnwidth}
        \centering
        \subfloat[\footnotesize Max-Cut instance with HIGH barren level (graph connectedness: 100\%, estimated barrenness: 79\%)]{\includegraphics[width=1\linewidth,height=0.5\linewidth]{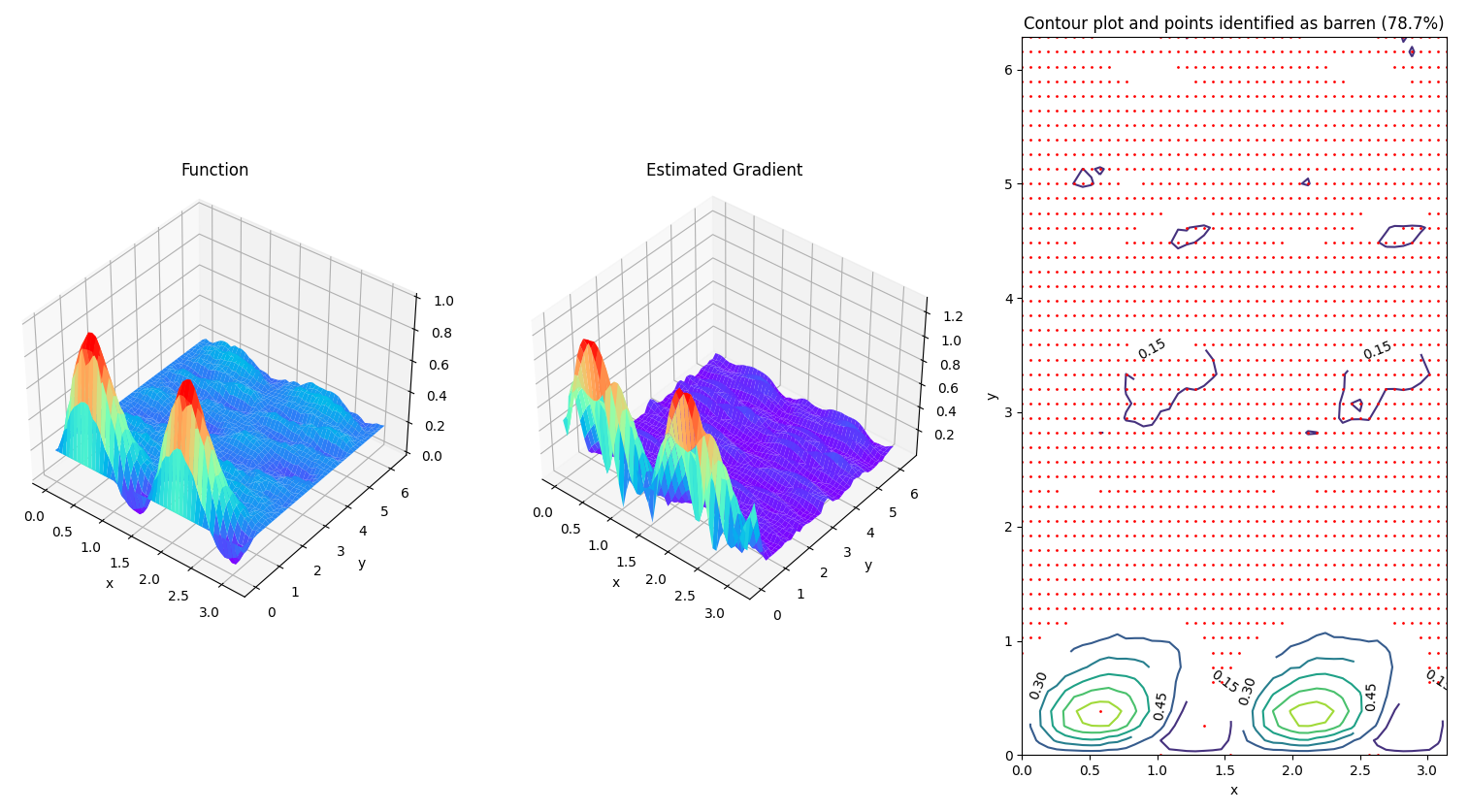}}
    \end{minipage}
    \caption{Examples of Max-Cut instances assigned to each of the three possible barren levels: (a) LOW, (b) MEDIUM, (c) HIGH. Left: function landscape; center: gradient norm based on 9x9 sliding window; right: contour plot of function with overlaid red points indicating areas identified as barren. All functions are evaluated on a $50 \times 50$ regular grid.}
    \label{fig:results_maxcut_instances}
\end{figure}

\section{Method to estimate the barren percentage}
\label{app:actual_barren_percentage_estimation}
The proposed method to estimate the percentage of barren areas in a 2D function consists of the following steps:
\begin{enumerate}
    \item   Define a regular grid on the parameter space, of say 100 points in each parameter dimension. We call this the evaluation grid.
    \item Classify each point in the evaluation grid as part of a barren area or not.
    \item Estimate the barren percentage as the percent of points classified as part of a barren area.
\end{enumerate}

Clearly, step (2) is the crucial step of the method. We have devised two classification methods, depending on the level of smoothness of the analyzed function, as follows:

\begin{enumerate}[label={\alph*)}]
    \item if the function is fairly smooth, as is the case with the functions used for the synthetic experiments presented in the paper, the gradient norm is computed separately at each grid point and assigned to it.

    %\item if there is inherent noise, an estimate of the noise variance, $\sigma^2$, is initially computed; then the local variance of the function in a predefined rectangular neighbourhood of each grid point (sliding window) is computed and compared against $\sigma^2$. The point is classified as part of a barren area when the F-test statistic associated to the ratio of the variances adjusted by their degrees of freedom is larger than a predefined threshold.
    \item if the function shows roughness, as is the case with the Max-Cut instances of the QAOA experiments, the average gradient vector is computed on a sliding window and its norm assigned to the grid point at the center of the window. The logic behind this approach is that, when rather noisy functions are to be optimized, barrenness is mostly due to a small and noisy gradient, which makes the optimizer move in very different directions but essentially stay around the same area.
\end{enumerate}

In both cases, a grid point is classified as part of a barren area when its assigned gradient norm is smaller than a threshold derived from its distribution over all grid points. Since there is no clear statistical test that can be used to choose the threshold, we chose it heuristically as the right-end of the first bin of the histogram of the gradient norm constructed on $\lfloor\sqrt{n}/4\rfloor$ bins, where $n$ is the number of points in the evaluation grid. The number of histogram bins was chosen by trial and error as a fraction of the square root of the number of points and by inspection of the results of the barren area classification process.

Figs.~\ref{fig:results_synthetic_instances} and \ref{fig:results_maxcut_instances} show the results of the barren areas classification on three function instances of each of the two optimization problem analyzed in this work; each instance corresponds to a different estimated barren level, as defined in the main text (see Section~\ref{sec:num_exp_synthetic}): LOW, MEDIUM and HIGH. We observe that the classification process yields reasonable results, in the sense that points classified as part of barren areas generate cohesive areas and their location tends to be consistent with contour lines.

% Use the following \parbox command or not depending on the page layout.
% The goal is to have the two figures showing the landscape of the Synthetic instances and the Max-Cut instances on the same page, for easier comparison / visualization.
% Ref: https://tex.stackexchange.com/questions/83930/what-are-the-different-kinds-of-boxes-in-latex
\parbox[c][\linewidth][c]{\linewidth}

%============== APPENDIX

\bibliographystyle{IEEEtran}
\bibliography{main-quantum}% Produces the bibliography via BibTeX.

\end{document}